\documentclass[aps,pra,notitlepage,onecolumn,nofootinbib,superscriptaddress,longbibliography,10pt]{revtex4-2}

\usepackage{amsmath,amsfonts,amssymb,array,graphicx,mathtools,multirow,bm,times,tcolorbox,relsize,booktabs}
\usepackage[utf8]{inputenc}
\usepackage[T1]{fontenc}
\usepackage{qcircuit}

\usepackage{mathrsfs}
\usepackage[ruled, vlined, linesnumbered]{algorithm2e}

\usepackage{hyperref}
\hypersetup{colorlinks=true,citecolor=blue,linkcolor=blue,filecolor=blue,urlcolor=blue,breaklinks=true}

\usepackage[marginal]{footmisc}
\usepackage{url}
\usepackage{theorem}

\usepackage{orcidlink}

\newtheorem{definition}{Definition}
\newtheorem{proposition}{Proposition}
\newtheorem{lemma}[proposition]{Lemma}

\newtheorem{theorem}[proposition]{Theorem}

\newtheorem{assumption}{Assumption}
\newtheorem{remark}{Remark}

\newenvironment{proof}{\noindent \textbf{{Proof~} }}{\hfill $\blacksquare$}

\allowdisplaybreaks

\usepackage{cleveref}
\usepackage{graphicx}
\usepackage{xcolor}

\RequirePackage[framemethod=default]{mdframed}
\newmdenv[skipabove=7pt,
skipbelow=7pt,
backgroundcolor=darkblue!15,
innerleftmargin=5pt,
innerrightmargin=5pt,
innertopmargin=5pt,
leftmargin=0cm,
rightmargin=0cm,
innerbottommargin=5pt,
linewidth=1pt]{tBox}
\definecolor{darkblue}{RGB}{0,76,156}

\newmdenv[skipabove=7pt,
skipbelow=7pt,
backgroundcolor=darkred!15,
innerleftmargin=5pt,
innerrightmargin=5pt,
innertopmargin=5pt,
leftmargin=0cm,
rightmargin=0cm,
innerbottommargin=5pt,
linewidth=1pt]{rBox}
\definecolor{darkred}{RGB}{195,0,0}
\definecolor{colortwo}{rgb}{0.4,0.77,0.17}
\definecolor{colorthree}{rgb}{0.01,0.51,0.93}

\DeclareMathOperator{\sgn}{sign}
\newcommand{\rmd}[0]{{\mathrm{d}}}
\newcommand{\iie}[0]{\textit{i.e.}}
\newcommand{\diag}[0]{{\rm diag}}
\newcommand{\para}[0]{\mathbin{\!/\mkern-5mu/\!}}

\begin{document}

\title{Adaptive controllable architecture of analog Ising machine}

\author{Langyu Li~\orcidlink{0009-0001-0781-5160}}\email{langyuli@zju.edu.cn}
\affiliation{College of Control Science and Engineering, Zhejiang University, Hangzhou 310027, China}

\author{Ruoyu Wu}
\affiliation{College of Information Science and Electronic Engineering, Zhejiang University, Hangzhou 310027, China}

\author{Yong Wang}
\affiliation{Hangzhou Innovation Institute, Beihang University, Hangzhou 310051, China}

\author{Guofeng Zhang}%
\affiliation{Department of Applied Mathematics, The Hong Kong Polytechnic University, Hong Kong, China}

\author{Jinhu L{\"u}}
\thanks{Corresponding authors:
\href{mailto:lvjinhu@buaa.edu.cn}{lvjinhu@buaa.edu.cn};
\href{mailto:gaoqing@buaa.edu.cn}{gaoqing@buaa.edu.cn};
\href{mailto:ypan@zju.edu.cn}{ypan@zju.edu.cn}}
\affiliation{School of Automation Science and Electrical Engineering, Beihang University, Beijing 100191, China}
\affiliation{Hangzhou Innovation Institute, Beihang University, Hangzhou 310051, China}

\author{Qing Gao}
\thanks{Corresponding authors:
\href{mailto:lvjinhu@buaa.edu.cn}{lvjinhu@buaa.edu.cn};
\href{mailto:gaoqing@buaa.edu.cn}{gaoqing@buaa.edu.cn};
\href{mailto:ypan@zju.edu.cn}{ypan@zju.edu.cn}}
\affiliation{School of Automation Science and Electrical Engineering, Beihang University, Beijing 100191, China}
\affiliation{Hangzhou Innovation Institute, Beihang University, Hangzhou 310051, China}

\author{Yu Pan}
\thanks{Corresponding authors:
\href{mailto:lvjinhu@buaa.edu.cn}{lvjinhu@buaa.edu.cn};
\href{mailto:gaoqing@buaa.edu.cn}{gaoqing@buaa.edu.cn};
\href{mailto:ypan@zju.edu.cn}{ypan@zju.edu.cn}}
\affiliation{College of Control Science and Engineering, Zhejiang University, Hangzhou 310027, China}

\begin{abstract}
As a quantum-inspired, non-traditional analog solver architecture, the analog Ising machine (AIM) has emerged as a distinctive computational paradigm to address the rapidly growing demand for computational power. However, the mathematical understanding of its principles, as well as the optimization of its solution speed and accuracy, remain unclear. In this work, we for the first time systematically discuss multiple implementations of AIM and establish a unified mathematical formulation. On this basis, by treating the binarization constraint of AIM (such as injection locking) as a Lagrange multiplier in optimization theory and combining it with a Lyapunov analysis from dynamical systems theory, an analytical framework for evaluating solution speed and accuracy is constructed, and further demonstrate that conventional AIMs possess a theoretical performance upper bound. Subsequently, by elevating the binarization constraint to a control variable, we propose the controllable analog Ising machine (CAIM), which integrates control Lyapunov functions and momentum-based optimization algorithms to realize adaptive sampling–feedback control, thereby surpassing the performance limits of conventional AIMs. In a proof-of-concept CAIM demonstration implemented using an FPGA-controlled LC-oscillator Ising machine, CAIM achieves a $2\times$ speedup and a 7\% improvement in accuracy over AIM on a 50-node all-to-all weighted MaxCut problem, validating both the effectiveness and interpretability of the proposed theoretical framework.
\end{abstract}

\date{\today}
\maketitle

\section{Introduction}

Due to the increasing demand for both general-purpose and specialized computational power, Ising machines, as dedicated solvers for quadratic unconstrained binary optimization (QUBO) problems, have made remarkable progress in recent years. Since QUBO problems are NP-complete in computational complexity theory, their complexity grows exponentially. For general-purpose computing, brute-force methods that search through the QUBO solution space typically require exponential time and computational resources. In contrast, Ising machines, designed specifically for this purpose, can overcome the algorithmic complexity limitations of the von Neumann architecture. They achieve analog solutions to QUBO problems based on different physical principles and have already demonstrated experimental advantages on several benchmark tasks \cite{johnson2011quantum,venturelli2015quantum,hamerly2019experimental,Chowdhury2023Full}. With the rapid advancement of Ising machines, many real-world combinatorial optimization problems are now modeled as QUBO instances and efficiently solved using these machines. Gradually, this has formed a complete workflow. Ising machines have proven effective across various fields, including finance, machine learning, and biomedicine, and continue to expand in both problem scope and scale—becoming an indispensable computational paradigm in specialized computing \cite{Lucas2014Ising,mcmahon2018solve,hussain2020optimal,bao2021multi,parizy2022cardinality,weinberg2023supply}.

Among various types of Ising machines, analog Ising machines (AIMs) have recently attracted widespread attention as an important branch. Ising machines can be realized through multiple physical implementations and solving principles. Alongside AIMs, there are quantum annealers (QAs) based on superconducting qubits \cite{johnson2011quantum,venturelli2019reverse,inoue2021traffic} and quantum annealing, coherent Ising machines (CIMs) built upon photonics \cite{wang2013coherent,inagaki2016coherent,bohm2018understanding}, and digital Ising machines (DIMs) focusing on hardware-level and parallel architectures \cite{yamaoka201524,yamaoka201520k,yamamoto2020statica,aadit2022massively,Wang2025Parallel,Li2025Fast}. AIMs, leveraging mature CMOS technology, have lower environmental requirements compared to QAs and CIMs, making them more suitable for near-term deployment and practical applications. Although DIMs also use CMOS technology, they rely on digital computation and primarily emphasize hardware acceleration. Consequently, their performance depends heavily on algorithmic design. AIMs, on the other hand, perform analog computation via Lyapunov processes, granting them both theoretical and implementation advantages, making AIMs a crucial branch in the Ising machine landscape.

Contribution to improve and understand the architecture and solving principles of AIMs have become a focal point, as shown in Fig. \ref{fig: architecture}a. Early AIMs were implemented using electronic oscillators, where inductor-capacitor (LC) oscillators are coupled to form a Kuramoto model with sub-harmonic injection locking (SHIL)\cite{wang2017sub,wang2017oscillator,wang2019oim,chou2019analog,Albertsson2023Highly}. This synchronization mechanism allows each oscillator to exhibit two stable states, corresponding exactly to the up and down spins of the Ising model. The system dynamics, i.e., the synchronization process of oscillators, can be expressed by the Adler equation with SHIL, which can be proven to follow a Lyapunov process \cite{adler1946study,acebron2005kuramoto}. Therefore, the system naturally evolves toward the lowest energy state, corresponding to the Ising model’s ground state. Subsequent studies such as bistable resistively-coupled Ising machine (BRIM) \cite{Afoakwa2021BRIM,Zhang2022CMOS} and ring oscillator-based AIMs (ROSC) \cite{lo2023ising,Yin2024Ferroelectric,Qian2025Device,Kumar2025DROID} introduced special bistable elements to achieve more controllable coupling and smoother convergence. Meanwhile, many works have incorporated machine learning and other optimization techniques to enhance AIM architectures \cite{bohm2022noise}. Other studies examined the solving principles from different perspectives, including the impact of SHIL strength, noise, and perturbations on system dynamics \cite{Erementchouk2022computational,Bashar2023Stability}. Inspired by this, several AIM variants have introduced stochastic or deterministic perturbations, such as those based on simulated annealing or Langevin dynamics, to improve solution accuracy and convergence speed \cite{chou2019analog,Liu2025Integrated,Wu2025Enhancing}.

However, despite the various realizations of AIMs today, their Lyapunov function derivative is not strictly negative and often becomes zero during evolution. Since current AIMs are typically autonomous systems, the derivative being zero leads the system to premature stagnation or local minima, rather than the true global minimum. Moreover, whether the lowest energy state strictly corresponds to the ground state of the Ising model remains underexplored. As a result, there is still no unified theoretical framework addressing the reachability of system dynamics and the correctness of solutions. Furthermore, because autonomous AIMs cannot be externally controlled during solving, their achievable precision is inherently limited by this lack of controllability. These issues are crucial to the further development of AIMs, motivating this work to focus on providing theoretical insights and practical solutions.

\begin{figure}
    \centering
    \includegraphics[width=\textwidth]{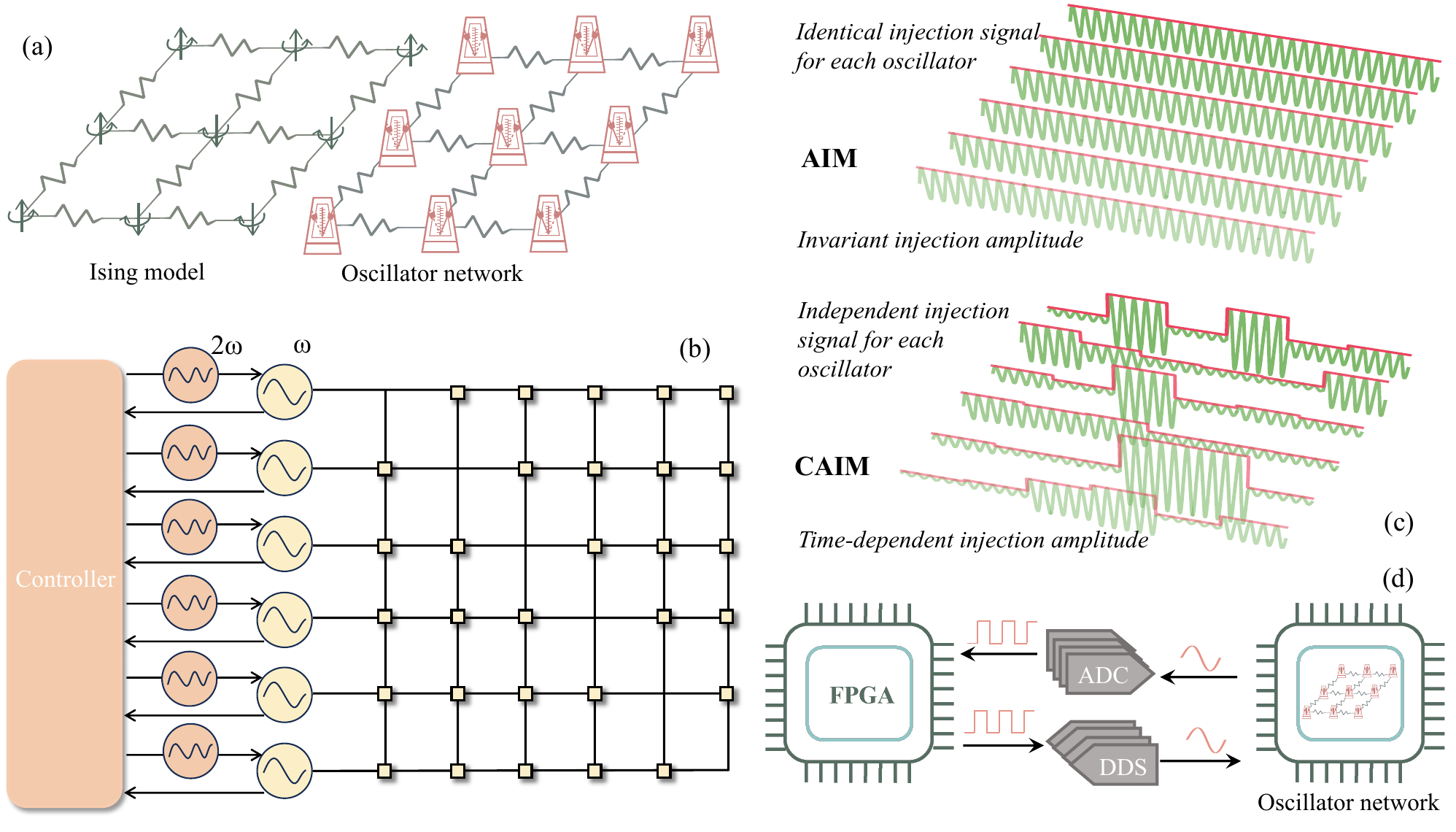}
    \caption{Proposed CAIM architecture and implementations. (a) A lattice-shaped Ising model and the corresponding oscillator network, where each spin in the Ising model has only two discrete states, while the states of the oscillator network are typically continuous and are usually represented by $0$ and $\pi$ as the two extrema. (b) The CAIM architecture. Compared with a conventional AIM, an additional controller is introduced and independent injection signals can be applied to each oscillator. The controller samples the oscillator states and computes feedback injection signals, thereby enabling precise regulation of the oscillator states, intervening in the intrinsic dynamical evolution of a conventional AIM, and consequently achieving faster and more accurate solution processes. (c) In a conventional AIM, identical injection signals with fixed strength are applied to all oscillators, whereas in a CAIM, different injection signals are applied to different oscillators and their injection strengths are adjusted in real time. (d) A proof-of-concept physical implementation framework of a COIM is presented, in which the controller is realized using an FPGA to perform signal sampling, computation, and feedback.}\label{fig: architecture}
\end{figure}

In this work, from the perspective of dynamical systems, we abstract previously proposed AIMs and establish a correspondence between oscillator bistability and the constraints in optimization theory. Based on this abstraction, we propose a systematic and unified theoretical analysis framework that can characterize the relationship between AIM dynamics and optimization performance, such as solution accuracy. Under this framework, we theoretically predict that traditional AIMs cannot simultaneously guarantee both solution correctness and reachability, which explains certain performance limitations observed in practice.

Building upon this theoretical understanding, we propose a controllable analog Ising machine (CAIM) architecture by introducing controller into AIMs, as shown in Fig. \ref{fig: architecture}bc. Specifically, SHIL strength is used as the control variable, and a digital partial controller is employed to achieve controllability. Furthermore, by integrating control Lyapunov functions with a momentum-based adaptive controller, we design an asynchronous adaptive control injection algorithm, which enhances both the reachability and correctness of solutions—thereby breaking previous performance limits while accelerating convergence.

Experimentally, we implement CAIM on a 51-node fully connected FPGA–OIM oscillator architecture, as shown in Fig. \ref{fig: architecture}d, achieving an average $2\times$ speedup and 7\% success probability improvement on benchmark problems such as weighted Max-Cut (MaxCuts) and Spin models (SpinModels). Additionally, through numerical simulations on benchmark problems, we analyze the effects of asynchrony, initial conditions, noise, hyperparameters, and AIM types on CAIM performance, validating its accuracy, speed improvement, and general applicability.

\section{Results}\label{sec: results}
\subsection{Formulation of AIM}
The Ising model was originally proposed to explain the phenomenon of ferromagnetic phase transition \cite{ising1925contribution}. The Hamiltonian of $n$-spin Ising model is defined as
\begin{align}
    H(\bm{s}) =&\sum_{i,j} J_{ij} s_i s_j + \sum_{i} h_i s_i \\
    =&\bm{s}^\top \bm{J} \bm{s} + \bm{h}^\top \bm{s},   \quad \bm{s}\in \{-1 ,+1 \}^n, \label{eq: ising hamiltonian}
\end{align}
where $\bm{J}:= [J_{ij}]_{n\times n}, \bm{h} :=[ h_1 , \cdots ,h_n ]^\top,\bm{s}:=[ s_1 , \cdots ,s_n ]^\top$, and each spin $s_i = \pm 1$ represents the two possible states (up and down) of a spin in the Ising model.
The coupling coefficients $J_{ij} \in \mathbb{R}$, satisfying $J_{ij} = J_{ji}$ and $J_{ii} = 0$, describe the interaction strength between spins $i$ and $j$.
The term $h_i \in \mathbb{R}$ denotes the external magnetic field applied to spin $i$.
If $\bm{J}$ is regarded as the pairwise interaction among spins, then $\bm{h}$ can be viewed as the individual bias for each spin.

The problem of finding the ground-state energy of the Ising model, known as the Ising problem, is equivalent to the QUBO formulation. Its objective is to find a spin configuration $\mathbf{s}^*$ that minimizes the Ising Hamiltonian, i.e.,
\begin{align}
    \bm{s}^* = \arg \min_{\bm{s}} \quad \bm{s}^\top \bm{J} \bm{s} + \bm{h}^\top \bm{s}, 
\end{align}
where $\mathbf{s}^*$ is also referred to as the optimal solution. Sometimes, the set of optimal solutions may contain more than one vector, but the Ising problem generally does not pay much attention to finding all of them. Therefore, we use $\bm{s}^*$ to denote one certain of them.

Inspired by the most representative LC oscillator-based Ising machine (OIM) in AIM, we attempt to generalize the concept and definition of AIM. The general model of the AIM is determined by 5 elements: the representation energy function $K:\Omega \rightarrow \mathbb{R}  $, the binarization energy function (or called regularization term) $R:\Omega \rightarrow \mathbb{R}$, the decision function $F:\Omega \rightarrow \mathscr{S}$, the weight of regularization $\mu \in \mathbb{R}^+$, and the dynamics $ \frac{\rmd \psi}{\rmd t}$, where $\psi:= [\phi_1,\cdots,\phi_n]^\top \in \Omega \subseteq \mathbb{R}^n$ is the state vector describing the each oscillator state of AIM during synchronizing, and $ \Omega$ denotes the state space; $\mathscr{S}:=\{ -1,+1\}^n$. Obviously, the state vector is time-dependent, \iie, $\psi (t)= [\phi_1 (t),\cdots,\phi_n (t)]^\top$, and the total energy function $E(t)$ of the system can be expressed as 
\begin{align}
 E(\psi(t)) &= K(\psi(t)) + \mu R(\psi(t)).   
\end{align}
The representation function $K(\psi)$ depends on the physical implementation of $\bm{J}$ and $\bm{h}$, realizing an approximation of the Ising Hamiltonian through the emulation of spin interactions. The binarization function $R(\psi)$ can be regarded as a regularization term or constraint, whose minimum typically enforces two stable states for each component of the state vector, often symmetric and corresponding to the two spin states in the Ising model. The system dynamics determine the trajectory of the state evolution, while the decision function maps the state vector to a solution of the Ising problem, corresponding to the readout process of the final AIM state.

To illustrate this more intuitively, two typical examples help clarify the mechanism. For the OIM, the representation function is $K_{\mathrm{OIM}}(\psi) = \sum_{i,j} J_{ij} \cos(\phi_i - \phi_j) + \sum_i h_i \cos(\phi_i),$
and its binarization term is $R_{\mathrm{OIM}}(\psi) = -\sum_i \cos(2\phi_i)$,
which enforces each oscillator to have two stable phase states, $0$ and $\pi$, corresponding to the lowest energy configuration. Referred to the implementation, $K(\psi)$ is obtained by the Kuramoto model of LC oscillators, and  $R(\psi)$ is implemented by SHIL, where $\mu$ is corresponding to the injection amplitude of SHIL. Under the condition that $R_{\mathrm{OIM}}(\psi)$ is minimized, $K_{\mathrm{OIM}}(\psi)$ becomes equivalent to the Ising Hamiltonian, since $\cos(\phi_i - \phi_j) = \cos(\phi_i)\cos(\phi_j) = s_i s_j$. Similarly, for the BRIM, the binarization term is $    R_{\mathrm{BRIM}}(\psi) \approx \sum_i (\phi_i^2 -1 )^2$. It ensures that, at its minimum, the representation function $ K_{\mathrm{BRIM}}(\psi) = \sum_{i,j} J_{ij} \phi_i \phi_j + \sum_i h_i \phi_i$, becomes equivalent to the Ising Hamiltonian as well. Unlike OIM, whose SHIL strength can be adjusted independently, BRIM has a fixed $\mu$, which is integrated into its oscillators. The dynamics of both OIM and BRIM follow
\begin{align}
\frac{\rmd \phi_i}{\rmd t} = -\frac{\partial E}{\partial \phi_i},~ (\text{or written as  } \dot{\psi} = -\nabla E(\psi) ,)
\end{align}
which guarantees $E(\psi(t))$ a Lyapunov process due to $\dot{E}(\psi(t)) = \langle \nabla E(\psi) , \dot{\psi} \rangle = -\Vert \nabla E(\psi)\Vert_{\ell_2}^2\leq 0$, driving the total system energy toward its minimum; $\langle \bm{a}, \bm{b} \rangle = \bm{a}^\top \bm{b}$ denotes the inner product, and $\nabla$ is the gradient operator; $\Vert \cdot\Vert_{\ell_2}$ denotes $\ell_2$-norm. Additionally, both models employ a decision function based on the sign operation $F_{\mathrm{OIM}}(\psi) = [\cdots, \operatorname{sign}(\cos(\phi_i)), \cdots]^\top, F_{\mathrm{BRIM}}(\psi) = [\cdots, \operatorname{sign}(\phi_i), \cdots]^\top$.

In simple terms, Lyapunov dynamics in AIM allow the system to minimize its energy while satisfying the regularization constraints, producing two stable equilibrium states.
These two states evolve, under energy minimization, toward the configuration that minimizes the representation energy function, which corresponds precisely to the Ising Hamiltonian, thereby solving the Ising problem.
Hence, we can summarize two essential conditions for an AIM to correctly solve the Ising problem: 
\begin{enumerate}
    \item The system energy must satisfy the Lyapunov condition (energy minimization).
    \begin{align}
    \dot{E}(\psi(t)) \leq 0,
    \end{align}
    \item When the binarization energy reaches its minimum, the representation energy equals the Ising Hamiltonian
\begin{align} \label{eq: simulation condition}
    K(\psi ) = H (F(\psi)), \quad \forall \psi \in \{ \psi :  R(\psi) =  \min  R(\psi)\}.  
\end{align}
\end{enumerate}

More discussion about this formulation can be found in the SI \ref{si-sec: formulation}. Generally, $R,F$ are always related to sign function. Therefore, for mathematical convenience, some small assumptions are made on the properties of $K,R,F$, which are always satisfied in the most existing AIMs. Generally, $K(\psi), R(\psi)$ is differentiable. $R$ is sign-independent for every element, \iie, $ R( [\cdots, \phi_i,\cdots  ]) =  R( [\cdots, -\phi_i,\cdots  ])$. $F$ is element-wise operator, \iie, $F([\phi_1 ,\phi_2 ,\cdots]^\top) =  [ F(\phi_1 ),F(\phi_2)\cdots]^\top$.

\subsection{Correctness of the AIM}

However, the system energy satisfies $\dot{E}(\psi(t)) \leq 0$ but is not strictly less than zero. Moreover, since the system is evidently autonomous, i.e., $\frac{\partial E(\psi, t)}{\partial t} = 0$, it implies that the system may get trapped at a fixed point (a local minimum), where the system state ceases to evolve, even though it is not the global minimum of the energy function. Therefore, the system may not always be able to find the correct solution, which requires a discussion on the actual performance of AIM.

The correctness of the solution, also referred to as the accuracy, is one of the most important performance metrics. It describes how close the Ising Hamiltonian corresponding to the system's minimum-energy state $\psi_* := \arg \min E(\psi)$ is to the true ground-state energy $H_0:=  \min H(\bm{s})$. Here, the concept of energy levels $H_0,H_1,\cdots$ is introduced, as for arbitrary $\bm{s}$, $H(\bm{s}) \in \{ H_0, H_1, \cdots \}, H_0< H_1<\cdots$. Based on the energy levels, the set of spin configurations associated with the $i$-th energy level can be defined as $\mathcal{S}_i:= \{ \bm{s} : H(\bm{s}) = H_i) \}, |\mathcal{S}_i| \geq 1$. Furthermore, the inverse of the decision function is defined as $F^{-1}(\bm{s}) := \{ \psi : F(\psi) = \bm{s} \}, F^{-1}(\mathcal{S}) := \{ \psi :  F(\psi) = \bm{s}, \; \bm{s} \in \mathcal{S} \}$.

In the ideal case, the solving principle of AIM can be formulated as a primal optimization problem
\begin{align}
    \arg \min_{\psi}& \quad K(\psi), \label{eq: primal problem 1}\\
    \text{subject to}& \quad R(\psi) = \min  R(\psi). \label{eq: primal problem 2}
\end{align}
Evidently, this primal problem is equivalent to the Ising problem according to Eq. \eqref{eq: simulation condition}, except that their variables lie in different state spaces, namely $\Omega$ and $\mathscr{S}$, respectively. 
By applying the method of Lagrange multipliers, an augmented Lagrangian function can be obtained
\begin{align}
    K(\psi) + \mu (R(\psi) -  R_{\min} ).
\end{align}
After discarding constant terms $- \mu R_{\min}$, the remaining part corresponds to the system energy of AIM, whose decreasing process reflects the optimization process itself. 
It is apparent that the augmented Lagrangian function represents a relaxed optimization problem. To ensure equivalence to the original optimization problem, the penalty coefficient $\mu$ must be sufficiently large, or even $\mu \to \infty$. This also implies that, in practical AIM implementations, the obtained solution may not necessarily coincide with the exact solution of the primal optimization problem. 

Therefore, it becomes necessary to introduce the definition of the equivalent solution. We say the relaxed problem has the equivalent solution as the primal problem when satisfying
\begin{align}\label{eq: equivalent solution}
    F(\psi_*) \in \mathcal{S}_0.
\end{align}
The purpose of this definition is to evaluate whether the theoretically minimal-energy state obtained by AIM actually corresponds to the ground state of the intended Ising problem. 
Since the constraint imposed by $R(\psi)$ is relaxed, the corresponding spin configuration space $\bm{s}$ may be expanded from $\mathcal{S}$ to $\mathbb{R}^n$. Within this enlarged space, a mismatch between the obtained solution and the true Ising ground state may occur. Therefore, a necessary and sufficient condition for holding equivalent solution can be formulated
\begin{align}
   \min_{\psi \in F^{-1}~ (\mathcal{S}_0)} E( \psi) \leq \min_{\psi  \in F^{-1} (\mathscr{S}/\mathcal{S}_0)}~ E( \psi),
\end{align}
which requires that at least one spin configuration belonging to $\mathcal{S}_0$ corresponds to the state with the minimal system energy. Further discussions on the definition of the equivalent solutions and their necessary and sufficient conditions can be found in the SI \ref{si-sec: correctness}.

This necessary and sufficient condition remains valid for various complex AIM architectures. However, in order to build a more intuitive connection among different AIM elements, several mathematical assumptions are introduced based on this condition. Since Eq. \eqref{eq: simulation condition} describes the situation in which the representation energy coincides with the Ising Hamiltonian under the given constraint, and practical AIMs evolve continuously, we propose a stronger assumption: suppose that the strong constraint imposed by $\mu R(\psi)$, $\mu$ larger than a certain value $\bar\mu$, can restrict the minimum-energy state within the neighborhood $\mathscr{U}_{\bm{s}}$ of its extremum $\psi_{\bm{s}} :=\arg \min_{\psi \in F^{-1} (\bm{s}) } R(\psi)$, \iie,
\begin{align}
   \psi_\diamond \in \mathscr{U}_{\bm{s}},
\end{align}
 for $\forall \mu \geq \bar{\mu},\psi_\diamond := \arg \min_{\psi \in F^{-1}(\bm{s})} E(\psi)$. Moreover, $R(\psi)$ is treated as a metric function that measures the degree of approximation between the representation energy and the Ising Hamiltonian, expressed by their absolute energy difference. We further assume their consistency, that is, a smaller $R(\psi)$ indicates a closer correspondence between the representation energy and the Ising Hamiltonian, \iie,
\begin{align}
     \Big| K(\psi_1)-H(\bm{s})\Big|\leq \Big| K(\psi_2)-H(\bm{s})\Big|  ,
\end{align}
for $\psi_1,\psi_2 \in \mathscr{U}_{\bm{s}},R(\psi_1) \leq R(\psi_2)$. A detailed justification of this assumption can be found in the SI \ref{si-sec: correctness}.

Under this assumption, we can assert that there exists a sufficiently large $\widehat{\mu}$ that satisfies the equivalent-solution condition. Furthermore, for all $\mu > \widehat{\mu}$, the equivalent-solution condition continues to hold. 
By analogy with the relationship between the annealing time $\mathcal{O}(\frac{1}{(H_1-H_0)^2})$ in QA and the minimum energy gap $H_1- H_0$ \cite{farhi2000quantum}, the parameter $\widehat\mu$ is also inversely proportional to the energy gap as 
\begin{align}
    \widehat\mu\sim \mathcal{O}\Big(\frac{1}{H_1 - H_0}\Big).
\end{align}
Proofs can be found in the SI \ref{si-sec: correctness}.

This conclusion indicates that the existence of equivalent solutions depends on the magnitude of $\mu$. Specifically, this dependency arises because the neighborhood around the extremum of $R(\psi)$ introduces a relaxation effect. Although $R(\psi)$ attains its extremum, $K(\psi)$ may not necessarily do so—it could instead correspond to a saddle point or a trivial stationary point. If the gradient of $K(\psi)$ is sufficiently large, the system state vector can move toward a lower-energy point by slightly relaxing the constraint imposed by $R(\psi)$. However, a larger $\mu$ effectively suppresses this tendency by penalizing such deviations, thereby stabilizing the system near $\bm{s}$ and preserving the equivalence of the solution. This can be explained in more detail by Fig.\ref{fig: theoretical}a. Hence, $\mu$ exhibits a positive correlation with the correctness of the solution.

Regarding the interpretation of the energy gap, errors in the obtained solution occur when $K(\psi)$ fluctuates within a tolerable range defined by $R(\psi)$. In such cases, the system energy tends to converge near the first excited state. Therefore, a larger energy gap facilitates clearer separation between the ground and excited states, allowing for smaller values of $\mu$. Conversely, when the Ising problem has a very small energy gap, these two levels become difficult to distinguish, and even slight relaxations may cause the system to converge to the excited state. Hence, $\mu$ must be sufficiently large to confine the system energy within a narrow range and prevent incorrect solutions.

\subsection{Reachability and practical performance }\label{sec: reachability}

However, when discussing the practical performance of AIM, a larger $\mu$ does not necessarily lead to better results. The above conclusions are based on the assumption that the AIM dynamics can naturally evolve toward the global minimum-energy state. In practice, however, such an ideal condition is almost impossible to achieve. In dynamical systems, a fixed point is defined as a state $\psi$ such that $\dot{\psi}=0$, while a critical point of the energy function $E(\psi)$ is defined by $\nabla E(\psi)=0$. Once the state arrives at a fixed point, the state will no longer change. Since AIMs are typically autonomous systems with dynamics of the form $\dot{\psi} = -\nabla E(\psi)$, the fixed points and critical points coincide. Therefore, the dynamical behavior of the system is closely tied to the energy function $E$, and the analysis of $E$ provides predictive insights into system performance.

Because $R(\psi)$ serves as a binarization function, for each $\bm{s}$, it typically possesses at least one minimum point, which is also an extremum whose Hessian matrix is positive semi-definite. For example, in the case of $R_{\mathrm{OIM}}(\psi)= -\sum_i \cos(2\phi_i)$, these minima correspond to the stable states $\psi = [\cdots, k_i \pi,\cdots]^\top, k_i \in \mathbb{Z}$; For $R_{\mathrm{BRIM}}(\psi)=\sum_i (\phi_i^2 -1 )^2$, its stable states can be $\psi = [\pm 1, \pm 1,\cdots]^\top$. Due to the convexity and smoothness around these points, they form regions that act as local minima, as shown in Fig. \ref{fig: theoretical}b. When minimizing the total energy $E(\psi)$, \iie, following the direction of $-\nabla R(\psi)$, the positive semi-definite Hessian of $R(\psi)$ attracts the state toward these extremum.

For larger $\mu$, the attraction gets stronger, making states trapped in a low-energy domain. For example, even if two states $\psi \in F^{-1} (\bm{s})$ and $\psi^\prime  \in F^{-1} (\bm{s}^\prime)$ have the same value of $\mu R(\psi) =\mu R(\psi^\prime)$, and $K(\psi^\prime) < K(\psi)$, the transition from $\psi$ to $\psi^\prime$ may require passing through an intermediate state $\psi^{\prime \prime}$ with a slightly larger $ R(\psi^{\prime \prime}) =R(\psi)+\epsilon $, but a sufficiently large $\mu$, $\mu \epsilon > K(\psi) - K(\psi^{\prime \prime})$. In this case, the transition path is prohibited by the dynamics $\dot{E}(\psi(t)) \leq 0$, due to a too large $\mu R(\psi^{\prime \prime})$ term, which acts as an energy barrier, making the energetic cost of this transition too high. As previously analyzed, AIMs cannot always realistically evolve to the global minimum-energy state.

In particular, when $\mu R(\psi) \gg K(\psi)$, the system preferentially evolves along the direction of $-\nabla R$, leading to premature convergence to a local minimum and stagnation. Moreover, these local minima are typically exponentially numerous. For instance, ignoring the periodicity of OIM and BRIM, there is at least one local minimum for each $\bm{s}$. As $\mu \to \infty$, the total energy satisfies $E (\psi)\approx\mu R(\psi)$. Evolving along the direction where $\dot{E}(\psi(t)) \leq 0$, one can predict that for each $\bm{s}$, there exists an invariant set $\mathcal{I}_{\bm{s}} \subseteq F^{-1}(\bm{s})$ around $\psi_{\bm{s}}$.The states in  $\mathcal{I}_{\bm{s}} $ satisfies local Lyapunov stability, due to $\dot{E}(\psi(t)) \approx \mu \dot{R}(\psi(t))  < 0$. This causes the system’s final state $\psi_\infty$ to converge near a local minimum $\psi_{\bm{s}}$. For a trivial decision function $F$, such as a sign-like function, all points within $\mathcal{I}_{\bm{s}}$ are mapped to the same value, which makes the AIM’s solving process meaningless, because the final solution $F(\psi_\infty)$ becomes completely equivalent to the spin configuration of the initial state.

Define the probability that the AIM can randomly select an initial state that evolves to the global minimum-energy state as the reachability $P$. Thus, as $\mu \to \infty$, leading to the invariant set filling almost its whole part of state space, \iie, $\mathcal{I}_{\bm{s}} \approx F^{-1}(\bm{s})$, the reachability $P$ decreases exponentially such that
\begin{align}
    \mu\rightarrow \infty, \quad P\sim \mathcal{O}\Big(\frac{1}{2^n}\Big),
\end{align}
meaning each initial state converges directly to its own local minimum. Only when the system happens to start in $F^{-1}(\bm{s}^*)$ does it reach the true ground state, but this probability approaches zero as $n$ increases, making it extremely unlikely to find the global minimum.

When evaluating the practical performance of AIM, that is, how close the final Ising Hamiltonian energy is to the true ground-state energy, both correctness and reachability must be considered simultaneously. Both are strongly dependent on the regularization weight $\mu$, but in opposite ways: correctness tends to increase with larger $\mu$, while reachability decreases. A small $\mu$ makes the AIM solve a relaxed version of the Ising problem, which may result in incorrect solutions. Conversely, a large $\mu$ improves correctness but causes the AIM to prematurely converge to one of exponentially many local minima without sufficiently minimizing $K(\psi)$. Therefore, the practical performance of AIM may be upper-bounded, making it difficult to achieve both correctness and reachability simultaneously, as shown in Fig. \ref{fig: theoretical}c. Of course, there exist fortunate cases where, for certain $\mu$ and initial state $\psi_0$, the final state $\psi_\infty$ happens to correspond to $\bm{s}^*$ even though neither the equivalence condition nor the reachability condition is strictly satisfied. Such coincidences, however, are generally rare and are not considered in this work.

\begin{figure}
    \centering
    \includegraphics[width=\textwidth]{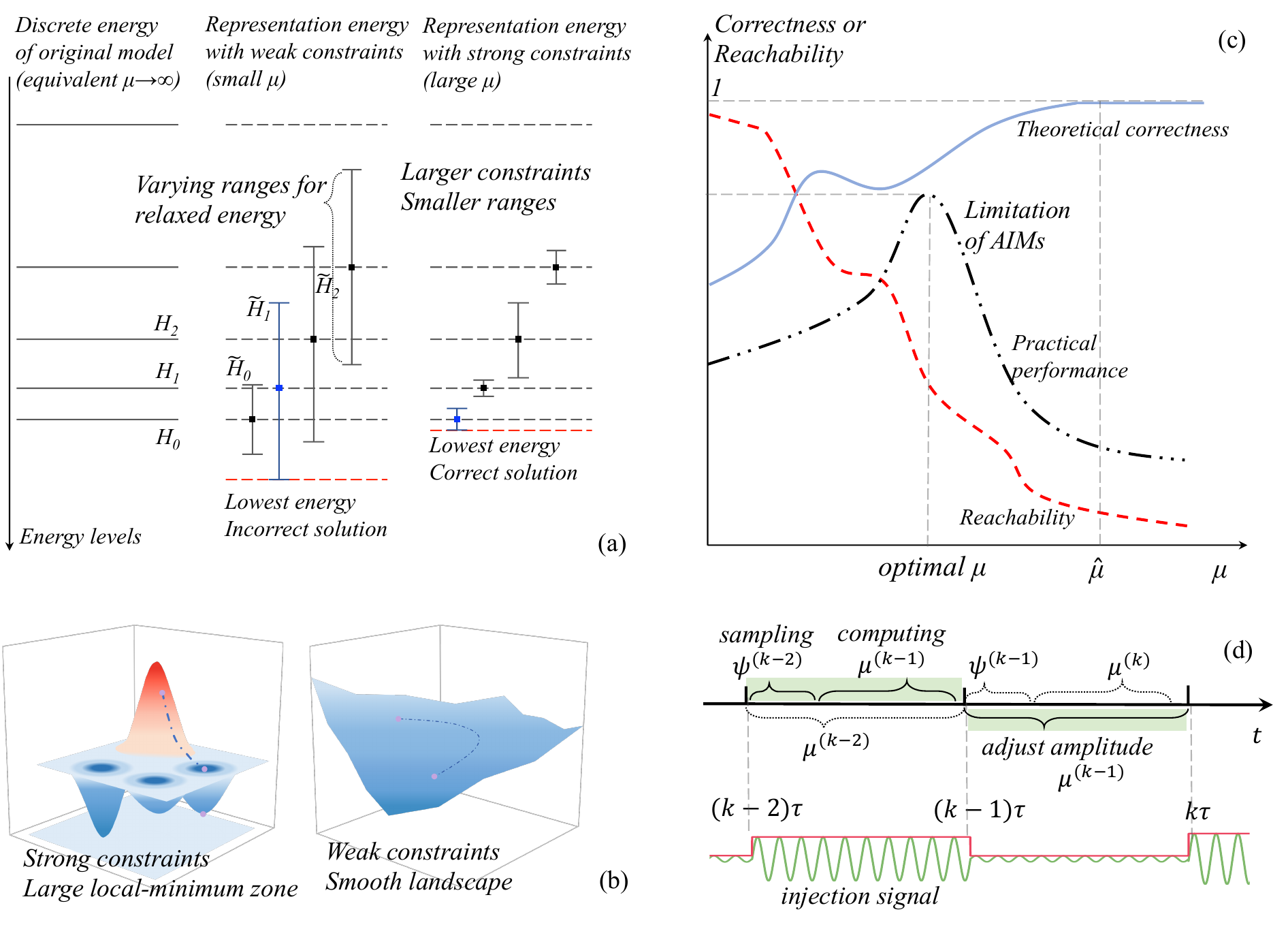}
    \caption{Theoretical limitation of AIM due to the trade-off of correctness and reachability. (a) The primal problem in Eqs. \eqref{eq: primal problem 1}$\sim$\eqref{eq: primal problem 2} can be regarded as having discrete energy levels, which is equivalent to the limit of regularization strength $\mu\to\infty$. For the relaxed problem, the state space is extended from discrete spin configurations $\bm{s}$ in the sets $\mathcal{S}_i$ (which yield Hamiltonian values $H_i$) to continuous representations $F^{-1}(\bm{s})$, such that the representation energy becomes continuous rather than discrete and spans a finite variation range. If, for some $\mathcal{S}_i$ with $i\neq 0$, the variation range of a certain spin is sufficiently large and its lower bound is the minimum among all variation ranges, then the theoretical solutions of the relaxed problem and the original problem are no longer equivalent. By increasing the regularization strength $\mu$, the variation range of the representation energy can be reduced, and when this range becomes sufficiently small, i.e., smaller than $H_1-H_0$, the solutions of the original and relaxed problems become equivalent. (b) An excessively large regularization strength $\mu$ enhances the attraction near the extrema of the regularization term, thereby forming local minima, and due to the non-increasing energy constraint of the Lyapunov process, once the system is trapped in a local minimum it cannot escape. Moreover, the regularization term typically introduces two extrema for each dimension, such as $\{\pm1\}$ or $\{0,\pi\}$, resulting in $\mathcal{O}(2^n)$ local minimum regions in practice. (c) Fig. (a) discusses solution equivalence from the perspective of optimization theory, while Fig. (b) addresses it from the perspective of dynamical systems; however, the former favors larger $\mu$, whereas the latter prefers smaller $\mu$, leading to a trade-off between correctness and reachability and consequently imposing an upper bound on the practical solution quality of AIMs. (d) Within each time slice of duration $\tau$, the controller first samples the oscillator states, then computes the amplitude according to the sampled states and the control algorithm, and outputs the result at the beginning of the next time slice to adjust the injection signal amplitude, and this process is repeated iteratively.
}\label{fig: theoretical}
\end{figure}

\subsection{Control architecture and algorithm design}\label{sec: control architecture design}

The previous discussion has analyzed the effect of $\mu$ on constraint strength, and the magnitude of $\mu$ plays a crucial role in determining system performance. For a time-independent $\mu$, the performance of AIM may be inherently limited, and finding an optimal $\mu$ is also challenging. Therefore, for a traditional AIM as an autonomous system, some external perturbation is required to overcome this limitation. Evidently, $\mu$ can be used as a control variable to manipulate the system state, making the establishment of a time-dependent $\mu(t)$ control scheme essential. Specifically, in the case of OIM, $\mu$ corresponds to the injection amplitude and can thus be tuned independently of the oscillator. However, for BRIM and ROSC, $\mu$ and $R$ are integrated into the oscillator components, making real-time adjustment during synchronization difficult. Here, we still assume that $\mu$ can be independently controlled by an external signal, as this is crucial for control architecture design and could be feasible in future implementations. Moreover, since control pulses in certain algorithms are not necessarily isotropic, $\mu(t)$ can be extended into a vector $\bm{\mu}(t)=[\mu_1(t),\mu_2(t),\cdots]^\top,~\mu_i (t) \in \mathbb{R}$, allowing heterogeneous control strength across oscillators. As shown in the SI \ref{si-sec: adaptive injection design}, such a vectorized $\bm{\mu}(t)$ controller enables full control of the system state, transforming the autonomous AIM into a controlled system, which is denoted as CAIM, whose dynamics are given by 
\begin{align}
\frac{\rmd \phi_i}{\rmd t} = -\frac{\partial K }{ \partial \phi_i} - \mu_i (t)\frac{\partial R}{ \partial \phi_i}.
\end{align}

Unlike previous studies, this work employs feedback control to construct the control framework. Feedback control operates in a closed-loop configuration, which generally achieves superior performance compared to open-loop control but also introduces higher computational complexity, power cost, and requires observable system states. In physical implementation, an additional analog or digital controller module is required alongside the AIM to perform the sampling–feedback process. Ignoring chip size and power constraints, a digital module (e.g., FPGA) allows easier compilation of control algorithms; however, for state-dependent algorithms that compute coupling-related intermediate variables, the computational complexity is typically $\mathcal{O}(n^2)$, implying that time delay must be carefully considered for high-speed AIMs. In contrast, an analog controller may achieve smaller time delay but is harder to implement in hardware. The general control model can thus be written as 
\begin{align}
    \mu_i (t) = \int_0^t \rmd t^\prime ~ f_i \Big(\psi(t^\prime)\Big).
\end{align}
Considering delay, since the control algorithm and hardware modules are relatively fixed, the delay remains approximately constant. Therefore, a sampling period $\tau$ slightly greater than the delay can be set as the sampling–feedback time interval, as shown in Fig. \ref{fig: theoretical}d. At the beginning of each period, the controller samples the oscillator states, performs control computation, and outputs control signals at the end of the period. Repeating this process forms a complete control chain. It is worth noting that the chain is asynchronous, since computation requires time and the oscillator states continue evolving during the control cycle. For example, the control algorithm proposed in this work depends on the previous two sampled states, and can be expressed as 
\begin{align}
        \frac{\rmd \phi_i}{\rmd t} =& -\frac{\partial K }{ \partial \phi_i} - \mu_i^{(k)} \frac{\partial R}{ \partial \phi_i}, \qquad \qquad\text{for } k\tau \leq t<(k+1)\tau,\\
    \bm{\mu}(t)\equiv &\bm{\mu}^{(k)} = f\big(\psi^{(k-1)},\psi^{(k-2)}\big),\quad \text{for } k\tau \leq t<(k+1)\tau .
\end{align}

Inspired by adaptive PID control optimization and momentum-based optimizers, which are equivalent to second-order Lyapunov systems capable of escaping local minima in non-convex optimization \cite{chen2024accelerated}, this work designs a control algorithm based on the control Lyapunov function (CLF) combined with momentum optimization. The CLF guarantees that system dynamics evolve along the gradient descent direction, satisfying Lyapunov conditions and improving convergence speed since it follows the steepest descent path. Meanwhile, the momentum optimizer mitigates the tendency of gradient descent to become trapped in local minima, thus overcoming traditional AIM limitations and improving accuracy. Combining both enables a balance between convergence rate and solution precision. Given an energy function $E^\prime (\psi) = K(\psi) + \mu^\prime R(\psi) $, the control Lyapunov function $\bm{\mu}(t)$ must satisfy
\begin{align}
     \langle \nabla E^\prime (\psi),  \dot{\psi} \rangle  =    ( \nabla E^\prime (\psi))^\top (  \nabla K (\psi) +  \nabla R (\psi) \cdot\bm{\mu}) \leq 0.
\end{align}
Then a direction can be derived to lead the fastest descent of energy
\begin{align}
    \bm{\mu} \propto  \nabla E^\prime (\psi) \cdot \nabla R (\psi) ,
\end{align}
which means that the CLF enforces gradient-descent-like dynamics, maintaining Lyapunov stability but converging rapidly into local minima. Second-order optimizers such as the momentum method or heavy-ball method can help escape such minima. Considering the asynchronous nature and computational delay, a discrete asynchronous momentum optimizer form can be derived such that
\begin{align} \label{eq: async momentum}
    \theta^{(t+1)} = \theta^{(t)} + \beta \Big( \theta^{(t-1)} - \theta^{(t-2)} \Big) - \eta \nabla \mathcal{J}(\theta^{(t-1})),
\end{align}
where $\mathcal{J}(\theta):\mathbb{R}^n \rightarrow \mathbb{R}$ denotes the objective function and $\{\theta^{(t)}\}$ the sequence of optimization variables; $\beta \in (0,1), \eta > 0$ is the decay ratio and learning rate. Interpreting each discrete iteration as one CAIM time step, the analytical control formulation of the proposed adaptive injection algorithm for CAIMs can be given as
\begin{align}\label{eq: adaptive algorithm}
 \bm{\mu}^{(k)} = -\beta \frac{1}{\nabla R(\psi^{(k-1)}) }\Big(\psi^{(k-1)} - \psi^{(k-2)}\Big)  +\eta \frac{\sqrt{n}}{\Vert E (\psi^{(k-1)}) \cdot\nabla R(\psi ^{(k-1)}) \Vert_{\ell_2}} \nabla E (\psi^{(k-1)}) \cdot\nabla R(\psi ^{(k-1)}),
\end{align}
where $\frac{1}{\nabla R(\psi ) }=[ \frac{1}{\frac{\partial R}{\partial \phi_1 }},\frac{1}{\frac{\partial R}{\partial \phi_2 }},\cdots]$ and $\frac{\sqrt{n}}{\Vert E (\psi) \cdot\nabla R(\psi) \Vert_{\ell_2}} $ is a normalization factor. Hence, the system dynamics can be viewed as: within each time slice, the asynchronous momentum optimizer determines the global energy descent direction, while intra-slice dynamics are governed by the natural evolution of AIM oscillators. This effectively achieves adaptive adjustment of $\mu$, surpassing the limitation of conventional AIMs with time-independent $\mu$. 

The CAIM control algorithm renders system dynamics approximately equivalent to an asynchronous momentum optimizer. Gradient-based optimizers are usually analyzed under convex or smooth assumptions, since theoretical characterization in non-convex scenarios is difficult. Unlike conventional asynchronous optimizers, the gradient term $- \eta \nabla \mathcal{J}(\theta^{(t-1}))$ in Eq. \eqref{eq: async momentum} is also asynchronous due to time delay, not synchronous as $- \eta \nabla \mathcal{J}(\theta^{(t}))$ \cite{Li2024Ordered}. If $\mathcal{J} ( \theta)$ is convex and $L$-Lipschitz continuous, $\nabla \mathcal{J} (\theta) $ is bounded $\Vert \nabla \mathcal{J} (\theta)  \Vert_{\ell_2} \leq C $,  then for $\beta \in [0,1), \eta \in (0, (1 - \beta )/L]$, the regret of the sequence $\{ \theta^{(t)} \}$ satisfies 
\begin{align}
\frac{1}{T}\sum_{t=1}^T  \Big (\mathcal{J}(\theta^{(t)}) -\mathcal{J}(\theta^*) \Big)  \sim \mathcal{O}  \Big(\frac{\eta^2 L(1+\beta)}{2(1-\beta)^2}C^2  \Big)  .
\end{align} 
The cumulative error thus scales as $\mathcal{O}(T)$, in contrast to $\mathcal{O}(1)$ for synchronous momentum optimizers \cite{ghadimi2015global}, which automatically eliminate steady-state error. The asynchronous version maintains a constant steady-state deviation from the optimum, which neither vanishes nor diverges. Therefore, asynchrony introduces an irreducible steady-state error that can be mitigated by parameter tuning but not fundamentally eliminated. Moreover, excessive delay in the CAIM is equivalent to a larger $\mu$, which will increase error and may lead to instability. Hence, CAIM delay must remain within a reasonable range and be minimized to ensure stability and convergence. Discussion of the special case without delay can be found in the SI \ref{si-sec: adaptive injection design}.

\subsection{Experimental demonstration of CAIM with 51-node FPGA-OIM}

\begin{figure}
    \centering
    \includegraphics[width=\textwidth]{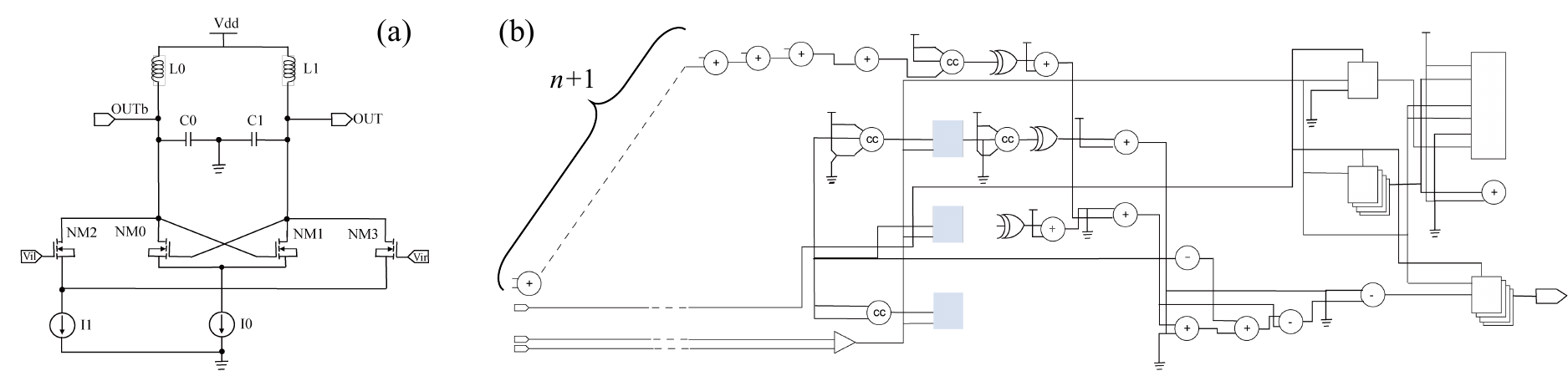}
    \caption{Oscillator unit and FPGA workflow design. }\label{fig: circuit}
\end{figure}

\begin{figure}
    \centering
    \includegraphics[width=\textwidth]{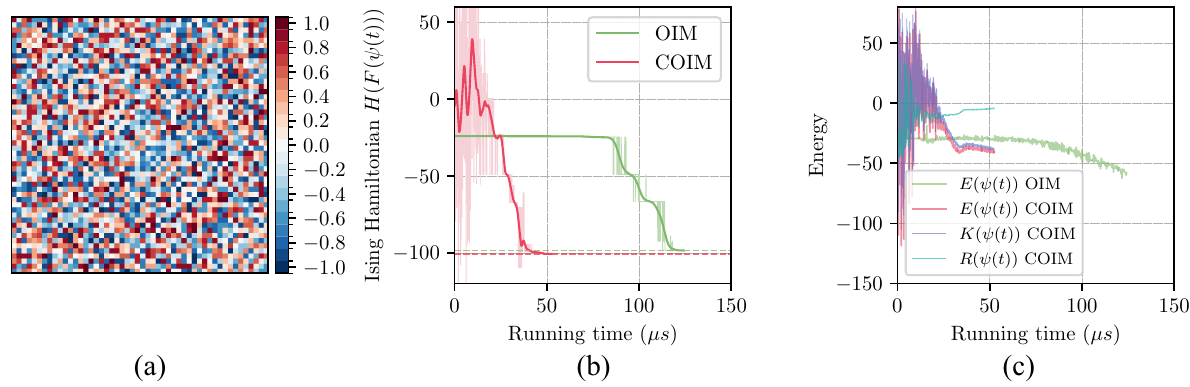}
    \caption{Comparison of COIM and OIM circuit experimental results. (a) An instance of a weighted MaxCut problem (SpinModel) with 50 nodes, where each row and column contains 50 grid points and the color indicates the coupling weight. (b) The real-time evolution of the Ising Hamiltonian during the solving process, where the light-colored curves denote raw energy values and the dark-colored curves represent smoothed results. COIM exhibits faster convergence and yields solutions closer to the ground state. (c) The evolution of different energy components during the solving process for OIM and COIM.
}\label{fig: compare}
\end{figure}

\begin{figure}
    \centering
    \includegraphics[width=\textwidth]{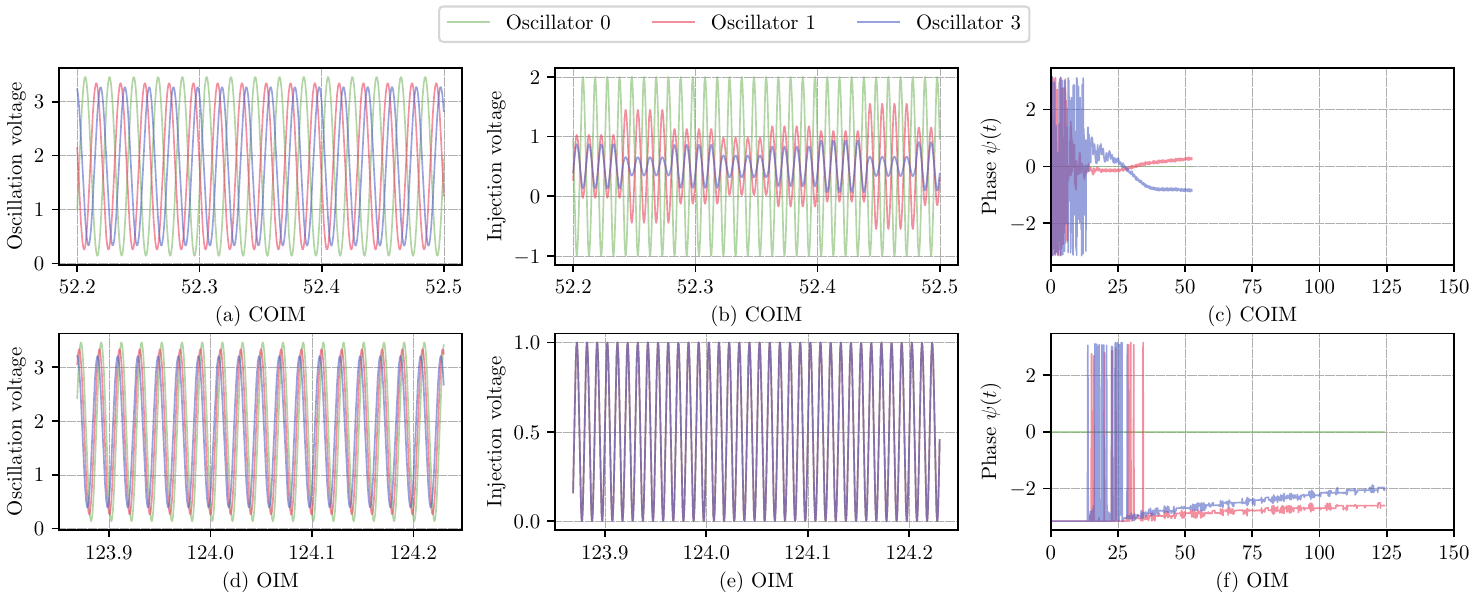}
    \caption{Differences between COIM and OIM during evolution (Time ($\mu$s), Voltage (V)). Since displaying all 51 oscillators would be redundant, oscillators 0, 1, and 3 are selected as representatives. (a) Oscillation waveforms of COIM oscillators near convergence. (b) The corresponding injection signals, showing that the injection amplitudes in COIM are time-varying and independently controlled for each oscillator. (c) The full-time phase evolution of COIM oscillators, whose rate of change is significantly higher than that of OIM. (d) Oscillation waveforms of OIM oscillators near convergence. (e) Injection signals of OIM, where the curves overlap because the amplitudes are constant and identical across oscillators. (f) The full-time phase evolution of OIM oscillators, which is noticeably slower than that of COIM.
}\label{fig: slidewindow}
\end{figure}

To experimentally validate the proposed Controlled AIM (CAIM) framework, we adopt the LC–oscillator based weighted Ising machine reported in \cite{chou2019analog} as the proof-of-concept physical substrate, as shown in Fig.~ \ref{fig: circuit}, with implementation details, experimental details and benchmarking presented in Sec. \ref{sec: methods}. The oscillator network consists of 51 LC oscillators to solve a fully connected Ising problem with bias (the $\bm{h}$ vector) on $n=50$ nodes. Each node employs identical components, with capacitance $C=4\,n\mathrm{F}$ and inductance $L=2.5\, n\mathrm{H}$, yielding a theoretical frequency of 50~MHz corresponding to a period of $T=20$~ns, and a DC bias set to 0.5~V. The controller is implemented using an FPGA (Xilinx Kintex-7 XC7K410T), with a reference injection amplitude of 0.5~V and a control delay set to $\tau=2T$; if the controller latency becomes excessive, the operating frequency can be reduced accordingly. The Ising problem is based on the SpinModel, which is equivalent to a weighted fully connected MaxCut problem, with parameters $J_{ij},h_i\in\pm0.1$ to $\pm1.0$; although $J$ and $h$ follow the same distribution, the effective coupling strength of $\bm{J}$ is twice that of $\bm{h}$ because $J_{ij}=J_{ji}$. In the oscillator network, the coupling strengths $J_{ij}$ are implemented using variable potentiometers. The 51 oscillator units are indexed from 0 to 50, where units 1 to 50 form the oscillator network based on $J_{ij}$, and unit 0 implements the bias $\bm{h}$ by coupling to the other oscillators with strengths $h_i$. Oscillator 0 plays a special role as a reference oscillator, and its specific functions in phase sampling and injection signaling are discussed in detail in Sec.~\ref{sec: methods}.

The experimental results are shown in Fig.~\ref{fig: compare}. Fig.~\ref{fig: compare}a presents a specific Ising problem, and the corresponding circuit experimental comparison is shown in Fig.~\ref{fig: compare}b. From the Hamiltonian decay curves in Fig.~\ref{fig: compare}a, COIM exhibits a clear advantage in convergence speed, achieving convergence in only $52.5\,\mu$s compared with $118.5\,\mu$s for OIM, corresponding to more than a twofold acceleration. Moreover, COIM finds solutions with lower energy and closer to the ground state, which typically implies a higher success probability of $1.37\times10^{-3}$ compared with $1.28\times10^{-3}$ for OIM, representing an improvement of approximately 7\%. Since accuracy improvements become increasingly difficult as solutions approach the ground state of a QUBO problem, this improvement is more clearly reflected in the time-to-solution (TTS), which is approximately 1.76~s for OIM and 0.43~s for COIM, corresponding to a 76\% reduction. These results demonstrate that COIM outperforms OIM in both accuracy and speed, achieving comprehensive performance improvements.

By combining Fig.~\ref{fig: compare}c and Fig.~\ref{fig: slidewindow}, the role of the adaptive injection algorithm can be understood more intuitively. As shown in Fig.~\ref{fig: compare}c, at the initial stage, when the oscillators are charging and starting to oscillate, the amplitudes are small and COIM has difficulty estimating real-time phases from subtle voltage differences, leading to pronounced oscillations in $E(\psi(t))$, $R(\psi(t))$, and $K(\psi(t))$. As the phases become more accurately estimated, the adaptive algorithm adjusts the injection amplitudes in real time based on the phases, as shown in Fig.~\ref{fig: slidewindow}b; according to the algorithm design, the adjusted injection locking strength drives the system along the direction of fastest energy descent, resulting in rapid energy reduction. As shown in Fig.~\ref{fig: slidewindow}c, the oscillator phases in COIM evolve much more rapidly than those in OIM. Additionally, in the tail region of Fig.~\ref{fig: compare}c, the regularization energy increases, indicating that the binarization constraint is relaxed at that stage; although the total energy of OIM at convergence is lower than that of COIM, the absence of constraint relaxation causes optimization of the binarization term to dominate the energy reduction, trapping OIM in a local minimum with a higher Ising Hamiltonian. In contrast, COIM relaxes the binarization constraint at appropriate times, enabling it to escape local minima and achieve a lower Ising Hamiltonian closer to the ground state, even if the total energy is not minimal. This phenomenon further validates the analysis of constraint strength presented in Sec.~\ref{sec: reachability}.

\section{Discussion}

In this work, we first present a generalized mathematical formulation of AIM, and by treating the constraint strength as a Lagrange multiplier, we establish an analytical framework that integrates dynamical systems theory and optimization theory to analyze the solving process, accuracy, and speed of AIM, and to reveal its theoretical performance limit. We then propose the CAIM framework, which treats the injection locking strength as a control variable, constructs a sampling–feedback control loop, and integrates asynchronous momentum methods with control Lyapunov functions to design an adaptive injection algorithm that adjusts the injection locking strength in real time, steering the solving process along the fastest descent direction while enhancing the ability to escape local minima. Finally, through comparative experiments on an FPGA-based LC-OIM platform, we validate the interpretability of the theoretical framework and demonstrate the superiority of the CAIM architecture and adaptive injection algorithm, particularly in terms of speed and accuracy improvements.

Several theoretical and practical assumptions are made, such as the adjustability of $\mu$ across different AIM implementations. Since control is realized by adjusting the injection amplitude, the amplitude is treated as a mathematical representation of constraint strength, and only the solving properties are analyzed; however, in physical circuits, amplitude is related to power and energy and affects the startup dynamics of oscillators, which is not considered here. Moreover, although the demonstrated controller is digital, future implementations may be analog and integrated with the oscillator network on the same chip. These aspects represent promising directions for future research.

In the Supplementary Information, more detailed mathematical descriptions, proofs, and explicit assumptions are provided. Furthermore, by abstracting away the physical circuit implementation, this work analyzes the CAIM architecture purely from a mathematical dynamical perspective and presents numerical experiments to investigate its properties with respect to injection strength, hyperparameters, and noise.

\section{Methods}\label{sec: methods}

\subsection{Phase sampling and injection feedback}

The FPGA controller performs phase sampling and input feedback via ADCs and direct digital synthesizer (DDS). The FPGA is connected to four groups of 8-bit ADCs for serial sampling, and since the ADC sampling delay is much smaller than the oscillator period, phase inaccuracies caused by serial sampling can be neglected. The oscillator phase cannot be obtained directly through sampling and must instead be computed, as sampling only provides direct circuit quantities such as current or voltage. By detecting the peaks (corresponding to phase 0) and valleys (corresponding to phase $\pi$) of the oscillation waveform and recording their times, linear extrapolation can be used to compute the phase at any time before the next peak or valley, denoted as 
\begin{align}
    \Delta\phi_i(t)=2\pi\frac{t-t_{i,k}}{t_{i,k+1}-t_{i,k}},
\end{align}
where $t_{i,k}$ and $t_{i,k+1}$ denote two consecutive detected extrema of the $i$-th oscillator. This method yields the real-time phase of the oscillator network, \iie, $\omega t+\phi_i$, whereas the control algorithm uses the characteristic phase $\phi_i$. Since all oscillators are oscillating, the oscillatory term $\omega t$ is eliminated by taking oscillator 0 as a reference, such that the phases of oscillators 1 to 50 are expressed as their real-time phases minus that of oscillator 0. Because oscillator 0 implements the bias $\bm{h}$ and must remain stable near $\pm1$, its injection locking signal is stronger and fixed at $4\mu$, independent of the control signal, such that its phase is $\omega t+0$ or $\omega t+\pi$, and subtracting this reference yields the characteristic phases of the other oscillators. Peak and valley detection requires voltage sampling at a frequency 20 times the oscillator frequency. After sampling, each oscillator uses three digital registers to store the current and two previous voltage samples, and comparisons are performed to identify peaks and valleys. To filter noise, the interval between detected peaks and valleys is required to exceed 0.4 times the period. Since the detected peak or valley occurs one sampling interval after the true event, the computed real-time phase lags by one sampling interval; however, the characteristic phase is unaffected because this lag is canceled when subtracting the reference oscillator phase.

The FPGA periodically adjusts the injection locking amplitude by computing the adaptive injection algorithm in Eq.~\eqref{eq: adaptive algorithm}. In Eq.~\eqref{eq: adaptive algorithm}, computing the $\ell_2$ norm and $\sqrt{n}$ requires square-root operations that consume substantial computational resources; since these terms serve only as normalization constants corresponding to Euclidean distance, the FPGA implementation replaces the $\ell_2$ norm with the $\ell_1$ norm and $\sqrt{n}$ with $n$, denoted as
\begin{align}
\frac{\sqrt{n}}{\Vert E (\psi^{(k-1)}) \cdot\nabla R(\psi ^{(k-1)}) \Vert_{\ell_2}} \approx  \frac{n}{\Vert E (\psi^{(k-1)}) \cdot\nabla R(\psi ^{(k-1)}) \Vert_{\ell_1}}.
\end{align}
In addition, to prevent excessive current from damaging transistors, an upper bound is imposed on the injection amplitude, set to $4\mu'$, such that when the absolute value of the control output exceeds $4\mu'$, the actual injected amplitude is clipped to $4\mu'$.

\subsection{Benchmarking}

When evaluating an Ising machine, the most important metrics are solution accuracy, solution speed, and the overall efficiency that balances both. Extensive prior work has discussed these aspects \cite{king2015benchmarking,hamerly2019experimental}, with the most representative metrics being the success probability $P$, defined as the probability that a single run reaches the ground state, the runtime $t_{\mathrm{run}}$ characterizing speed, and the time-to-solution (TTS) that accounts for both accuracy and speed, defined as $v = t_{\mathrm{run}} \log(1-0.99) / \log(1-P)$. However, the methods used to obtain $P$ and $t_{\mathrm{run}}$ require careful consideration.

For a given Ising instance, directly comparing the Hamiltonian values of different spin configurations allows straightforward evaluation of solution quality. However, energy comparisons across different problem instances are meaningless, even for instances of the same problem class. The success probability $P$ requires determining whether the obtained energy equals the ground-state energy, yet for high-dimensional Ising problems the ground-state energy is often unknown. Consequently, it is common to assume that a solution is successful if its energy is sufficiently close to the ground-state energy, i.e., within a tolerance $\varepsilon$. The choice of $\varepsilon$, however, cannot simultaneously accommodate different types of Ising problems, since practical applications often involve irregular and unstructured parameters rather than regular instances such as MaxCut. Therefore, cross-instance evaluation metrics are essential, and the ground-state approximation ratio serves as a useful alternative, $\frac{H_{\max} - H(\bm{s}) }{H_{\max} -H_{\min} }$, which quantifies how close a given energy level is to the ground-state energy. By bounding the energy using $H_{\min} \geq -n\Vert \bm{J}\Vert_\infty - \Vert h\Vert_1,H_{\max} \leq n\Vert \bm{J}\Vert_\infty + \Vert h\Vert_1 $ and $ \mathcal{O}(\Vert \bm{J}\Vert_\infty) \sim \mathcal{O} (\frac{1}{\sqrt{n}} \Vert \bm{J}\Vert_F)$, an estimated approximation ratio can be defined accordingly as
\begin{align}
    r=\frac{H_{\max} - H(\bm{s}) }{H_{\max} -H_{\min} }\approx \frac{1}{2} - \frac{H(\bm{s})}{2(\sqrt{n}\Vert \bm{J}\Vert_\infty + \Vert h\Vert_1 )}.
\end{align}

Because the computation of $P$ is not standardized \cite{hamerly2019experimental,lo2023ising}, and motivated by the assumption that the expected approximation ratio is positively correlated with success probability, here adopt an empirical numerical formula inspired by prior work to estimate $P$, given by \begin{align}
\widehat{P} = \exp (-\sqrt{n}(1.35- r) ) , 
\end{align}
where the coefficient $1.35$ is chosen as a conservative upper bound observed from numerical experiments on approximation ratios; in most cases, the observed values do not exceed $1.25$, and the chosen value ensures robustness around the scale of $1 \pm 0.2$. The factor $\sqrt{n}$ corresponds to the slope of this function on a logarithmic scale and serves as a heuristic balance for the degradation of approximation quality with increasing problem size.

The runtime $t_{\mathrm{run}}$ has often been used for cross-platform comparisons among solvers such as simulated annealing, CIMs, AIMs, and quantum annealers, where the runtime is typically predefined and does not require direct measurement. Unlike sampling-based iterative algorithms such as simulated annealing, AIMs are dynamical systems governed by a Lyapunov function representing energy, and the system evolution naturally approaches convergence. If a predefined runtime is too long, post-convergence time may be unnecessarily included, leading to an overestimation of $t_{\mathrm{run}}$. Thus determine $t_{\mathrm{run}}$ based on the convergence of the energy function, defining the runtime as the time at which the energy variation falls below a predefined threshold, as given by
\begin{align}
    \frac{\max \{E_\mathrm{seg}\}   - \min \{ E_\mathrm{seg} \} }{E_\mathrm{init}-E_\mathrm{now}} \leq 0.025, \quad\{ E_\mathrm{seg} \} :=\{ E(\psi(t)): t_\mathrm{now} - \varepsilon \leq t\leq t_\mathrm{now}\} .\label{eq: convergence detection}
\end{align}

Accordingly, the TTS is computed as 
\begin{align}
    v = t_{\mathrm{run}} \frac{\log(1-0.99)}{ \log(1-\widehat{P} )}.
\end{align}
It should be noted that $t_{\mathrm{run}}$ may represent either the physical time of a hardware circuit or a dimensionless mathematical phase time describing system evolution in state space. The correspondence between these two notions is determined by the circuit gain $A_s$ in the dynamical equation $\frac{\rmd \phi_i}{\rmd t} = -A_s \frac{\partial E}{\partial \phi_i}$. The parameter $A_s$ depends on circuit implementation details such as device characteristics and operating frequency, but from an experimental perspective it can be approximated as being inversely proportional to the signal period $T$, i.e., $A_s \sim \mathcal{O}(1/T)$. As a result, physical runtimes may differ by orders of magnitude across platforms. By abstracting away physical implementation details and setting $A_s = 1$, the phase time provides a normalized basis for analyzing AIM dynamics.

\bibliography{ref}

\section*{Supplementary information}
\appendix

\section{Formulation of AIM}\label{si-sec: formulation}

\begin{figure}[h]
    \centering
    \includegraphics[width=\textwidth]{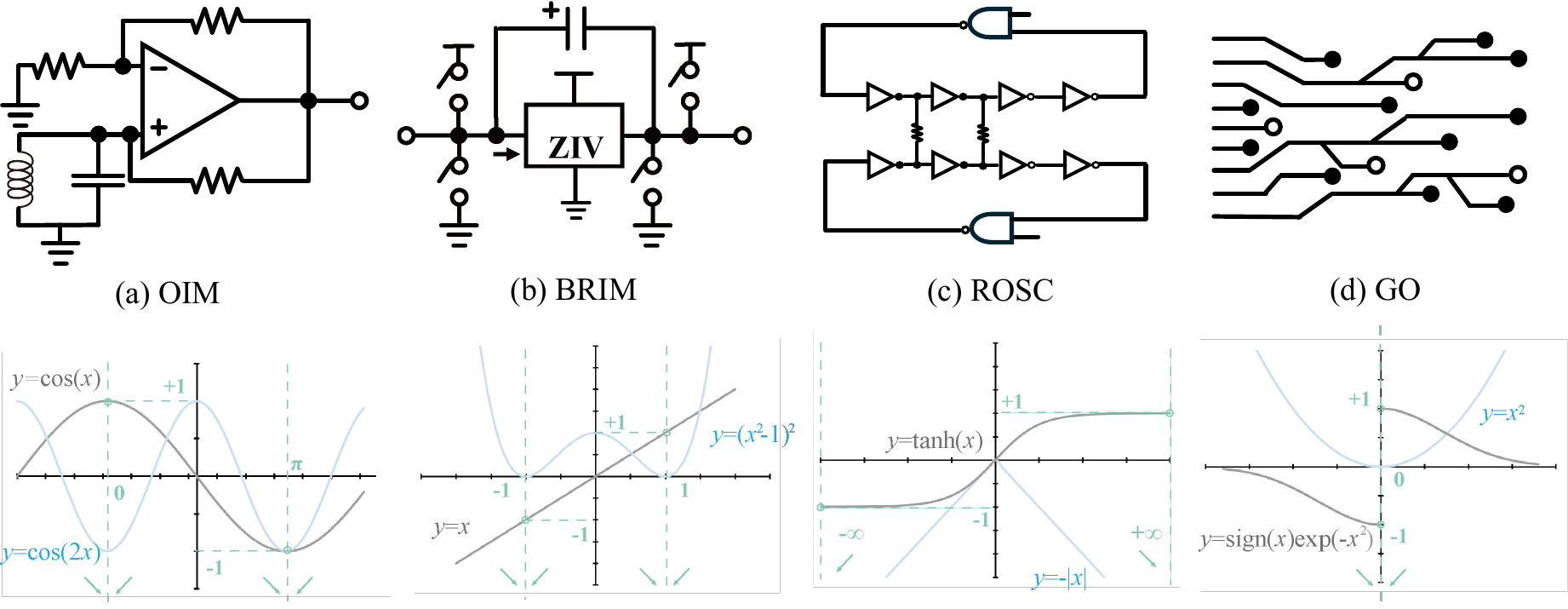}
    \caption{State-of-the-art AIM types and their modelings.}\label{si-fig: formulation}
\end{figure}

The implementation of AIM follows several representative architectures, including the LC–oscillator Ising machine (OIM), bistable resistively-coupled Ising machine (BRIM), and the ring-oscillator Ising machine (ROSC). Based on these realizations, a unified abstract model is summarized as presented in Def. \ref{si-def: aim}, and the corresponding mappings for each AIM variant are listed in Tab. \ref{si-tab: aim types}. It is worth noting that for BRIM, the function $R(\psi)$ is derived as a mathematical approximation; its practical form can be referenced from \cite{Afoakwa2021BRIM,Zhang2022CMOS}, and the approximation is sufficiently close to the actual behavior in terms of functional shape. For ROSC, the representation function is ideally modeled as a square wave and therefore expressed using the sign function $\sgn(\cdot)$, but such a formulation violates continuity and has zero gradient. Additionally, real physical voltages cannot switch instantaneously. Consequently, in our analysis we replace the sign function by $\tanh(\gamma \phi_i)$, and the ROSC representation becomes $K_{\mathrm{ROSC}} = \sum_{i,j} \tanh(\gamma \phi_i)\tanh(\gamma \phi_j)$. Since $\gamma \to \infty$ implies $\tanh(\gamma \phi) \to \sgn (\phi)$, and $\gamma |\phi| \to \infty$ implies $\tanh(\gamma \phi) = \pm 1$, we treat $\gamma \phi$ as a single binarization variable and use $\mu$ as the effective replacement of $\gamma$. This leads to the approximated binarization function $R(\psi) \approx -\sum_{i} |\phi_i|$, which achieves its minimum at infinity and corresponds to the upper and lower limits of a periodic pulse.

To investigate the generalization capability of the proposed control architecture, we additionally consider a hypothetical oscillator-based analog Ising machine that does not yet exist in practice, namely the Gaussian-augmented oscillator (GO). Its binarization function attains its sole extremum at zero. Unlike $R_{\mathrm{OIM}}$, $R_{\mathrm{BRIM}}$, and $R_{\mathrm{ROSC}}$, which all possess two extrema at $\pm 1$, the binarization function of GO contains only a single extremum, making $R_{\mathrm{GO}}$ smoother and preventing the system from becoming trapped in local minima. However, its representation function is discontinuous at zero, with left and right limits both having magnitude one but different signs $\pm 1$. Such a function exhibits convergence behavior different from the other AIM variants. Although this architecture has not yet been realized, it is plausible that similar systems may emerge in future research. Therefore, GO is included in the generalization discussion to highlight the potential capability of the proposed control framework. A summarization can be found in Fig. \ref{si-fig: formulation} and Tab. \ref{si-tab: aim types}.

\begin{definition}[Analog Ising machine]\label{si-def: aim}
The general model of the analog Ising machine is determined by 5 elements, the representation energy function $K:\Omega \rightarrow \mathbb{R}  $, the binarization energy function $R:\Omega \rightarrow \mathbb{R}$, the decision function $F:\Omega \rightarrow \mathscr{S}$, the weight of regularization $\mu \in \mathbb{R}^+$, and the dynamics $ \frac{\rmd \psi}{\rmd t}$, where it should satisfy the following conditions:
\begin{enumerate}
    \item Condition of Lyapunov function 
    \begin{align}
    E(\psi(t)):=&  K(\psi(t)) + \mu R(\psi(t)),\\
    \dot{E} (\psi(t)) =& \langle \nabla E , \dot{\psi} \rangle \leq  0;
    \end{align}
    \item Condition of equal energy under binarization
    \begin{align}
    K(\psi ) = H (F(\psi)), \quad \forall \psi \in \{ \psi :  R(\psi) =  \min  R(\psi)\}. \label{si-eq: equal binarized energy}
    \end{align}
\end{enumerate}
\end{definition}

\begin{table}[h]
\centering
\resizebox{\textwidth}{!}{

\begin{tabular}{c cccc}
\hline\hline
 & OIM&BRIM&ROSC&GO \\ \hline
 $K(\psi)$& $ \quad \sum_{i,j} J_{ij} \cos(\phi_i- \phi_j) \quad$ & $\quad\sum_{i,j} J_{ij} \phi_i \phi_j \quad$ &$\quad\approx\sum_{i,j} J_{ij} \tanh (\phi_i )\tanh(\phi_j) \quad$  & $\quad
     \sum_{i,j} J_{ij} \sgn (\phi_i ) \exp (-\frac{\phi_i^2}{2}) \sgn(\phi_j) \exp (-\frac{\phi_j^2}{2})\quad$\\
$R(\psi)$& $-\sum_i \cos(2\phi_i)$ & $\approx\sum_{i} (\phi_i^2-1)^2  $ & $ \approx - \sum_{i}|\phi_i|$  & $\sum_{i} \phi_i^2$\\
$\mu $& independent & integrated & integrated  & / \\
$F(\psi)$ & $\sgn (\cos(\phi_i))$ &  $\sgn (\phi_i)$ &  $\sgn (\phi_i)$  &  $\sgn (\phi_i)$ \\
$\frac{\rmd \phi_i}{\rmd t}$& $-\nabla E(\psi) $ & $-\nabla E(\psi) $ & $-\nabla E(\psi) $  & $-\nabla E(\psi) $ \\
\hline\hline
\end{tabular}
}
\caption{Typical AIM examples.}
\label{si-tab: aim types}
\end{table}

It is typically assumed that the decision function $F$ is trivial, such as a sign-based function that determines the final spin state purely by selecting the nearest binary value, since the sign function maps any real number to the closest of $\pm 1$. As this work does not focus on analyzing the influence of $F$, the decision function is always treated as trivial and is required not to access information beyond the system state. If one were to assume an $F$ that always outputs the optimal spin configuration $\bm{s}^\ast$ regardless of its input, such an overly powerful decision function would render all discussions about AIM performance meaningless and would not be realistic for AIM implementations. Existing AIM architectures also employ trivial decision functions. Furthermore, the oscillators in these AIMs are all interconnected according to the Kuramoto model, and their dynamics take the form $\dot{\psi} = -\nabla E(\psi)$.

The implementation of $\mu$ is the key to the proposed control architecture. In LC-oscillator-based OIMs, $\mu$ corresponds to the strength of injection locking and is independent of the oscillator network, allowing each oscillator to be independently changed with a time-varying control signal. In contrast, for BRIM and ROSC, the implementation of $\mu$ is embedded within the electronic components and therefore cannot currently be adjusted independently. However, it is plausible that equivalent mechanisms for modulating $\mu$ may be realized in future hardware. Since our goal is to analyze the general control framework independent of specific hardware constraints, we assume that $\mu$ is adjustable and omit physical implementation limitations in the theoretical discussion.

\section{Auxiliary analysis of the dynamical system of AIM}
\subsection{Correctness and reachability}\label{si-sec: correctness}

In discussing correctness, the concept originates from optimization theory and assumes an idealized AIM dynamics in which the system naturally evolves toward the global minimum-energy state $\psi^\ast$, which is also the intent of the Lyapunov process underlying AIMs. Therefore, $\psi^\ast$ is directly tied to the conclusions in this section, and correctness becomes independent of the system dynamics, since the process of reaching $\psi^\ast$ is idealized and thus ignored.

\begin{definition}[Equivalent solution]\label{si-def: equivalent solution}
Given an AIM defined in Def. \ref{si-def: aim} and an Ising problem $\bm{J},\bm{h}$, if the spin configuration of the minimum-energy state $\psi_*$ has the lowest Ising Hamiltonian, \iie, 
\begin{align}\label{si-eq: equivalent solution}
    F(\psi_*) \in \mathcal{S}_0,
\end{align}
we call the AIM has the equivalent solution with the corresponding Ising problem.
\end{definition}

Def.~\ref{si-def: aim} introduces the notion of equivalent solutions, meaning that if the theoretical minimum-energy point corresponds exactly to an optimal spin configuration, then the AIM solution is considered correct. This imposes constraints on $K$, $R$, $F$, and $\mu$. Based on this definition, a more concrete necessary and sufficient condition for equivalence can be formulated.

Def. \ref{si-def: equivalent solution} gives a definition on the correctness of the solution given by AIM. We give an necessary and sufficient condition on this.

\begin{proposition}[Condition of equivalent solution]\label{si-prop: condition of equivlent solution}
To satisfy the Def.~\ref{si-def: equivalent solution}, the AIM should satisfy
\begin{align}
   \min_{\psi \in F^{-1}~ (\mathcal{S}_0)} E( \psi) \leq \min_{\psi  \in F^{-1} (\mathscr{S}/\mathcal{S}_0)}~ E( \psi).
\end{align}
\end{proposition}
\begin{proof}
According to the Def.~\ref{si-def: equivalent solution}, Eq.~\eqref{si-eq: equivalent solution} is equivalent to
\begin{align}
&\exists \psi_1  \in F^{-1}~ (\mathcal{S}_0), \forall \psi_2 \in F^{-1} (\mathscr{S}/\mathcal{S}_0),\quad E(\psi_1)\leq E(\psi_2),\\
\Leftrightarrow & \min_{\psi \in F^{-1}~ (\mathcal{S}_0)} E( \psi) \leq \min_{\psi  \in F^{-1} (\mathscr{S}/\mathcal{S}_0)}~ E( \psi),
\end{align}
which completes the proof.
\end{proof}

Prop. \ref{si-prop: condition of equivlent solution} provides a strict necessary and sufficient condition ensuring equivalence, requiring that the minimum attainable energy of the ground-state subspace be strictly lower than that of any excited-state subspace. However, this condition is not intuitive. To obtain a more interpretable understanding of the role of $\mu$, several idealized assumptions about AIM must be introduced. Typically, $F$ is chosen as a sign-like function, and since there are $2^n$ possible spin configurations $\bm{s}$, the preimage $F^{-1}(\bm{s})$ also consists of $2^n$ subspaces of equal size. Moreover, $R$ is usually an even function and independent of sign, implying that $R$ exerts identical influence across all corresponding subspaces. Therefore, subsequent assumptions are made on these subspaces rather than globally across the entire state space.

\begin{assumption}\label{si-asp: asymptotic approch}
Given $K(\psi), R(\psi)$ of an AIM, with $\psi_{\bm{s}} :=\arg \min_{\psi \in F^{-1} (\bm{s}) } R(\psi) ,\psi_\diamond := \arg \min_{\psi \in F^{-1}(\bm{s})} E(\psi)$, suppose there exists a neighborhood $\mathscr{U}_{\bm{s}}$ of $\psi_{\bm{s}}$ for a certain $\bar{\mu} >0 $, satisfying
\begin{enumerate}
    \item $\forall \mu \geq \bar{\mu}$, it holds $\psi_\diamond \in \mathscr{U}_{\bm{s}}$;
    \item For $\psi_1,\psi_2 \in \mathscr{U}_{\bm{s}}$, if $R(\psi_1) \leq R(\psi_2) $, it holds $\Big| K(\psi_1)-H(\bm{s})\Big|\leq \Big| K(\psi_2)-H(\bm{s})\Big|$.
\end{enumerate}
\end{assumption}

Asp. \ref{si-asp: asymptotic approch} introduces the convergence region $\mathscr{U}_{\bm{s}}$ and describes an idealized property of AIM. Since $R(\psi)$ plays the role of a regularization term, its purpose is to attract the state toward its extrema. With sufficiently large regularization strength $\bar{\mu}$, the convergence region can be made small enough so that the minimum-energy point of each subspace lies within this neighborhood. Combined with the equality condition for binarized energies in Eq. \eqref{si-eq: equal binarized energy}, $R$ can be understood as the cost of violating the constraint in an optimization problem, and the assumption reflects an asymptotic consistency: the smaller the value of $R(\psi)$, the closer $\psi$ is to an extremum, and the closer the corresponding energy is to the true Ising Hamiltonian. This assumption is intuitive and consistent with the requirement imposed by Eq. \eqref{si-eq: equal binarized energy}.

\begin{theorem}
Given the AIM defined in Def.~\ref{si-def: aim}, with Asp.~\ref{si-asp: asymptotic approch}, there exists a large enough $\widehat\mu \geq \bar\mu$ satisfying the condition of equivalent solutions, and all $\mu \geq \widehat\mu$ satisfies the condition of equivalent solutions. 
\end{theorem}
\begin{proof}
With $\psi_\diamond := \arg \min_{\psi \in F^{-1}(\bm{s})} E(\psi)$ for a certain $\bm{s}$, it holds
\begin{align}
    K(\psi_\diamond ) \leq H(\bm{s}), 
\end{align}
because if $ K(\psi_\diamond ) > H_i$, there will be another $\psi$ satisfying $H_i< K(\psi )< K(\psi_\diamond ), R(\psi) <  R(\psi_\diamond)$, where $E(\psi) <E(\psi_\diamond ) $, which contradicts the definition of $\psi_\diamond$. This means that the final representation energy of the optimal state in the space of $F^{-1}(\bm{s})$ is always smaller than the Ising Hamiltonian.

For a larger $\mu_2,\mu_2 > \mu_1 \geq \bar\mu$, $\psi_{\diamond,2} :=\arg \min_{\psi \in F^{-1}(\bm{s})} K(\psi) + \mu_2 R(\psi)$, $\psi_{\diamond,1} :=\arg \min_{\psi \in F^{-1}(\bm{s})} K(\psi) + \mu_1 R(\psi)$, there holds
\begin{align}\label{si-eq: energy change}
    R(\psi_{\diamond,2})\leq  R(\psi_{\diamond,1}),  K(\psi_{\diamond,1} ) \leq K(\psi_{\diamond,2} ) \leq  H(\bm{s}), E(\psi_{\diamond,2})\geq  E(\psi_{\diamond,1}),
\end{align}
because of
\begin{align}
    &\left\{\begin{array}{c}
        K(\psi_{\diamond,2}) + \mu_2 R(\psi_{\diamond,2}) \leq  K(\psi_{\diamond,1}) + \mu_2 R(\psi_{\diamond,1})\\
        K(\psi_{\diamond,1}) + \mu_1 R(\psi_{\diamond,1}) \leq K(\psi_{\diamond,2}) + \mu_1 R(\psi_{\diamond,2}) 
    \end{array}\right. \\
    \Rightarrow  &  \mu_1  ( R(\psi_{\diamond,1}) -  R(\psi_{\diamond,2})  )   \leq  K(\psi_{\diamond,2}) -K(\psi_{\diamond,1}) \leq \mu_2 ( R(\psi_{\diamond,1}) -  R(\psi_{\diamond,2})  ). \label{si-eq: bounded derivative}
\end{align}
Therefore, given a larger $\mu $, $R(\psi_\diamond)$ will become smaller, and $H(\bm{s}) - K(\psi_\diamond) $ is smaller, where we can define
\begin{align}\label{si-eq: asympototic function of mu}
   f_{\bm{s}} (\mu)=&  K(\psi_\diamond)-H(\bm{s}) + \mu (R(\psi_\diamond)- R_{\min}),
\end{align}
in which $ f_{\bm{s}}(\mu)$ increases monotonically as $\mu$ grows. Finally, $f_{\bm{s}}(\mu)$ will approach $0^-$, and the convergence speed of $\mu R(\psi_\diamond)$ is slower than $K(\psi_\diamond)$, \iie, 
\begin{align}
   \frac{\partial K(\psi_\diamond) }{\partial \mu}\geq  -\frac{\partial \mu R(\psi_\diamond)}{\partial \mu}\geq0.
\end{align}

Next, the condition of equivalent solution can be rewritten as 
\begin{align}
     \min_{\bm{s} \in\mathcal{S}_0 }\min_{\psi \in F^{-1}(\bm{s})} E( \psi) &\leq  \min_{\bm{s} \in\mathscr{S}/\mathcal{S}_0 }\min_{\psi  \in F^{-1} (\bm{s})}~ E( \psi),\\
     \Leftrightarrow  ~~~~~~~~~~~~~~~\min_{\bm{s} \in\mathcal{S}_0 } E( \psi_{\diamond})  &\leq  \min_{\bm{s} \in\mathscr{S}/\mathcal{S}_0 } E( \psi_{\diamond}),\\
      \Leftrightarrow  ~~~~~~ \min_{\bm{s} \in\mathcal{S}_0 }  f_{\bm{s}} (\mu) + H_0 &\leq \min_{i=1,2,\cdots} \min_{\bm{s} \in\mathcal{S}_i } f_{\bm{s}} (\mu) + H_i,
\end{align}
where a sufficient condition of $\widehat\mu$ can be given $  H_0\leq\min_{i=1,2,\cdots} \min_{\bm{s} \in\mathcal{S}_i } f_{\bm{s}} (\mu) + H_i$. Therefore, the feasible $\widehat\mu$ satisfies
\begin{align}
    -f_{\bm{s}}(\widehat\mu) \leq H_i- H_0,\quad \forall i=1,2,\cdots, \forall \bm{s} \in \mathcal{S}_i,
\end{align}
which completes the proof with the monotonicity of $ f_{\bm{s}}(\mu)$.
\end{proof}

From the proof, it follows that increasing $\mu$ drives $R(\psi)$ to smaller values and pushes $K(\psi)$ closer to the exact Ising Hamiltonian. When $\mu$ becomes sufficiently large so that the relaxed energy does not exceed the first excited-state energy, incorrect solutions can no longer occur. Hence, larger $\mu$ more easily satisfies the equivalence condition because the regularization is stronger.

The proof also shows that the critical value $\widehat{\mu}$ depends on the energy gap between the first excited state and the ground state. A larger gap allows easier separation between energy levels and permits a smaller regularization weight. Thus, relating different Ising instances to the appropriate choice of $\mu$ helps clarify the mechanism of AIM. Since the neighborhood $\mathscr{U}_{\bm{s}}$ contains a monotonic function $f_{\bm{s}}$, and because both $K(\psi_\diamond)$ and $\mu R(\psi_\diamond)$ have similar convergence rates as shown in Eq. \eqref{si-eq: bounded derivative}, the assumption captures this structural similarity.

Because the derivative of $K(\psi_\diamond)$ is bounded in Eq.~\eqref{si-eq: bounded derivative}, where the derivative of $\mu R(\psi_\diamond)$ are included, indicating they are close.
\begin{proposition}
Suppose the scaling of $K(\psi_\diamond)-H(\bm{s}), \mu R((\psi_\diamond)-R_{\min})$ has the same order in $f_{\bm{s}}(\mu)$, \iie, $H(\bm{s}) -K(\psi_\diamond) = \alpha (\mu)\times \mu R((\psi_\diamond)-R_{\min})$, where $\alpha (\mu) \geq 1,\alpha(\mu) \sim \mathcal{O}(1) $. Then the $\widehat\mu$ has a scaling of 
\begin{align}
    \widehat \mu \sim \mathcal{O} \Big(\frac{1}{H_1 - H_0}\Big).
\end{align}
\end{proposition}

\begin{proof}
Given $\mu_2 \geq \mu_1$, there holds
\begin{align}
    f_{\bm{s}}(\mu_2) -  f_{\bm{s}}(\mu_1) = & K(\psi_{\diamond,2}) -K(\psi_{\diamond,1})+ \mu_1 (R(\psi_{\diamond,2}) -R(\psi_{\diamond,1})) + (\mu_2 -\mu_1) (R(\psi_{\diamond,2})-R_{\min})\\
    \leq & (\mu_2 -\mu_1) (R(\psi_{\diamond,1})-R_{\min}).
\end{align}
Let $\mu_1 = \mu, \mu_2 =\mu+\Delta \mu$, we have
\begin{align}
    \frac{\partial f_{\bm{s}}(\mu) }{\partial \mu}\leq&   \frac{1}{(1-\alpha(\mu))} \frac{f_{\bm{s}}(\mu)}{\mu },\\
    f_{\bm{s}}(\mu)\sim&  \mathcal{O}\Big( \frac{1}{(1-\alpha(\mu))} \frac{1}{\mu}\Big)\sim  \mathcal{O}( - \frac{1}{\mu}).
\end{align}
Additionally, $f_{\bm{s}}(\mu)$ has the same scaling for $\bm{s}$. When $\mu$ is large enough, its coefficients caused by different $\bm{s}$ become less significant, so the $H_1-H_0$ will be the worst case, which completes the proof.
\end{proof}

This proposition gives a relationship between the scope of $\widehat\mu$ and the minimum energy gap $H_1 -H_0$. In other words, the smaller the energy gap, the stronger the constraint needs to be. Otherwise, due to the loose constraints, the representation energy will have a larger variance from its true value, making it easier for the ground state to transition to the energy of an excited state. Therefore, stronger constraints make it easier to distinguish energy levels.

However, the theoretical $\widehat{\mu}$ concerns only the global minimum-energy point, assuming that the system can reach the true minimum of $E$. In reality, AIM dynamics may not be strictly decreasing, or even if they are, convergence time may be unbounded, making the theoretical assumption unattainable in practice. Therefore, it is necessary to analyze the system state during its practical convergence trajectory.

From Fig. \ref{si-fig: formulation}, it can be observed that for $R_{\mathrm{OIM}}$, $R_{\mathrm{BRIM}}$, and $R_{\mathrm{ROSC}}$, the region between the two extrema at $\pm 1$ contains an energy barrier, creating two separate basins of attraction. Because transitions during the Lyapunov process cannot increase energy, state transitions between basins become difficult. Increasing $\mu$ amplifies this restriction, further confining the system to the nearest basin determined by $R$.

As a result, although a $\mu$ has been provided that satisfies the condition for equivalent solutions, this condition is often difficult to meet in actual calculations. In practice, a $\mu$ that is not large enough cannot guarantee equivalent solutions, thus losing the theoretical guarantee of finding the global minimum. An appropriately chosen $\mu$ can make the equivalent problem sufficiently close to the Ising problem, and the output solution will be sufficiently close to the global optimum. Therefore, some correctness can be sacrificed, as an excessively large $\mu$ can cause $\mu$ to more easily fall into local minima.

\subsection{Adaptive injection algorithm design}\label{si-sec: adaptive injection design}

As mentioned above, for a time-invariant $\mu$, the practical performance of the conventional AIM always has a limitation. Therefore, a time-varying method can be designed to modify the control parameters. The time-independent case can also be included in the time-dependent case, which has a higher degree of freedom, making the system have more flexibility. Thus, if the time-invariant case is not always optimal, the time-varying case holds the potential to break the performance limitations in the Fig. \ref{fig: theoretical}c.

Therefore, the time-dependent control $ \bm{\mu}(t) := [\mu_1 (t),\cdots,\mu_n(t)]^\top,~\mu_i (t) \in \mathbb{R}$ can be designed, making dynamics become
\begin{align}\label{si-eq: control dynamics}
\frac{\rmd \phi_i}{\rmd t} = -\frac{\partial K }{ \partial \phi_i} - \mu_i (t)\frac{\partial R}{ \partial \phi_i}.
\end{align}

\begin{remark}\label{si-rmk: contrallability}
For element-independent $R(\psi)$, all dimensions are independently controllable, \iie, the transition from arbitrary state to another state is feasible, except the points with $\frac{\partial R}{ \partial \phi_i}=0$.
\end{remark}
This demonstrates the theoretical controllability of the system. If voltages and currents in physical implementation are not bounded, then as long as $\frac{\partial R}{\partial \phi_i} \neq 0$, an impulsive control signal can instantly move the system state. If only finitely many dimensions satisfy $\frac{\partial R}{\partial \phi_i} = 0$, the system may wait for natural dynamics induced by $E$ to escape, after which control becomes effective again. If the system falls into a fixed point but $\frac{\partial R}{\partial \phi_i} \neq 0$, adjusting $\bm{\mu}$ enables the system to leave that fixed point. However, if $\frac{\partial R}{\partial \phi_i} = 0$, the controller becomes ineffective at that moment.

To maintain Lyapunov stability, control over each dimension of $\bm{\mu}$ must be coordinated. Control Lyapunov functions (CLFs) are widely used in dynamical system analysis, and since Eq. \eqref{si-eq: control dynamics} resembles a CLF structure, a generalized CLF formulation can be applied to AIM \cite{Garg2020Prescribed}.

Given an energy function $E^\prime (\psi) = K(\psi) + \mu^\prime R(\psi)$ with a fixed $ \mu^\prime $, the generalized control Lyapunov function $\bm{\mu}(t)$ is the time-dependent control such that
\begin{align}\label{si-eq: generalized CLF}
  \langle  \nabla  E^\prime(\psi) ,   \dot{\psi} \rangle\leq 0,
\end{align}
in which the constraints of $\bm{\mu}$ are a half hyperplane, with $ \frac{\partial K}{\partial \psi}:= [\frac{\partial K}{\partial\phi_1},\cdots,\frac{\partial K}{\partial\phi_n}]^\top,  \frac{\partial R}{\partial \psi}:=[\frac{\partial R}{\partial\phi_1},\cdots,\frac{\partial R}{\partial\phi_n}]^\top$, such that
\begin{align}
    -( \nabla E^\prime (\psi))^\top ( \frac{\partial K}{\partial \psi} + \diag( \frac{\partial R}{\partial \psi} )\bm{\mu}) &\leq 0,\\
    \Leftrightarrow~~~~~~~~~~~~~~~~~~ -\nabla^\top  E^\prime (\psi)   \diag( \frac{\partial R}{\partial \psi} )\bm{\mu}&\leq  \nabla^\top  E^\prime (\psi)  \frac{\partial K}{\partial \psi}.\label{si-eq: half hyperplane}
\end{align}
The control pulse $\bm{\mu}$ can be expressed as
\begin{align}
    \bm{\mu} (t) = \big(\alpha (t) + \alpha^*(t)\big) \bm{\mu}_{\para}  (t) ,
\end{align}
where $\bm{\mu}_{\para} := \frac{1}{\Vert \diag( \frac{\partial R}{\partial \psi} )\nabla E^\prime (\psi) \Vert_{\ell_2}}  \diag( \frac{\partial R}{\partial \psi} )\nabla E^\prime (\psi) , \alpha^* (t):= \nabla^\top  E^\prime (\psi)  \frac{\partial K}{\partial \psi}$. Here, $\bm{\mu}_{\para}$ can be interpreted as the direction of steepest descent of the current energy, corresponding to the normal vector of a half-space, while $\alpha^\ast (t)  $ represents the minimum admissible step size. To make the injection amplitude is equal to the AIM's, \iie, $\Vert\bm{\mu}\Vert_{\ell_2} = \mu^\prime\Vert \bm{1}\Vert_{\ell_2}= \sqrt{n}\mu^\prime$, the coefficient of $\bm{\mu}_{\para}$ can be set as $\alpha (t) + \alpha^*(t)= \sqrt{n}\mu^\prime$. This direction accelerates convergence but does not alleviate the issue of falling into local minima; indeed, the CLF can be viewed as a strengthened Lyapunov system, increasing the likelihood of local-minimum trapping. Thus, additional controller design is required.

The CLF approximates the behavior of gradient-based optimizers, and momentum-based or heavy-ball methods can help escape local minima. By incorporating the asynchronous and delay analyzed in Sec. \ref{sec: control architecture design}, an asynchronous momentum optimizer can be formulated.

Given the objective function $\mathcal{J}(\theta)$, the update equation of the momentum method with one-step delay can be described as 
\begin{align} \label{si-eq: async momentum}
    \theta^{(t+1)} = \theta^{(t)} + \beta ( \theta^{(t-1)} - \theta^{(t-2)} ) - \eta \nabla \mathcal{J}(\theta^{(t-1)}),
\end{align}
where $\beta \in (0,1), \eta > 0$ is the decay ratio and learning rate.

Because AIM control is subject to delay, the controller’s output $\bm{\mu}$ is partitioned into discrete segments, making the AIM dynamics more similar to the discrete-time systems targeted by momentum methods. Since momentum algorithms possess memory, they demonstrate inherent robustness to delay. Convergence analysis can therefore be conducted in the convex setting, as analyzing non-convex convergence is generally difficult.

\begin{lemma}[Convergence of asynchronous momentum]
 If $\mathcal{J} ( \theta)$ is convex and $L$-Lipschitz continuous, $\nabla \mathcal{J} (\theta) $ is bounded $\Vert \nabla \mathcal{J} (\theta)  \Vert_{\ell_2} \leq C $,
then for then for $\beta \in [0,1), \eta \in (0, (1 - \beta )/L]$, the sequence $\{ \theta^{(t)} \}$ satisfies 
\begin{align}
\frac{1}{T}\sum_{t=1}^T  (\mathcal{J}(\theta^{(t)}) -\mathcal{J}(\theta^*))  \sim \mathcal{O} (\frac{\eta^2 L(1+\beta)}{2(1-\beta)^2}C^2)  ,
\end{align}
where $\theta^* := \arg \min \mathcal{J}(\theta)$.

\end{lemma}
\begin{proof}
First, $\mathcal{J}(\theta)$ is convex and $L$-Lipschitz continuous, so we have \cite{ghadimi2015global,kim2017convergence,nesterov2004introductory}
\begin{align}\label{eq: convex condition start}
0 \leq  \mathcal{J}(y) -  \mathcal{J}(x) - \langle \nabla  \mathcal{J}(x), y - x \rangle &\leq \frac{L}{2} \Vert x - y\Vert^2_{\ell_2}, \\
 \mathcal{J}(x) + \langle \nabla  \mathcal{J}(x), y - x \rangle + \frac{1}{2L} \Vert\nabla  \mathcal{J}(x) - \nabla  \mathcal{J}(y)\Vert^2_{\ell_2} & \leq  \mathcal{J}(y),\\
 \langle \nabla \mathcal{J}(y)-\nabla  \mathcal{J}(x), y - x \rangle&\leq \frac{L}{2} \Vert x - y\Vert^2_{\ell_2} .\label{eq: convex condition end}
\end{align}
And we find the $\Vert  \theta^{(t)} -\theta^{(t-1)} \Vert_{\ell_2}$ is bounded by
\begin{align}
   \Vert  \theta^{(t)} -\theta^{(t-1)} \Vert_{\ell_2} \leq  \beta \Vert  \theta^{(t-2)} -\theta^{(t-3)} \Vert_{\ell_2} + \eta \Vert  \nabla \mathcal{J}(\theta^{(t-1)}) \Vert_{\ell_2} \leq&  \beta \Vert  \theta^{(t-2)} -\theta^{(t-3)} \Vert_{\ell_2} + \eta C\\
   \leq& \beta (\beta \Vert  \theta^{(t-4)} -\theta^{(t-5)} \Vert_{\ell_2} +\eta C)+ \eta C\\ 
   \leq & \cdots = \eta C \frac{1-\beta^\frac{t}{2}}{1-\beta}.
\end{align}

Then we can rewrite the iteration equation
\begin{align}
    \theta^{(t+1)} - \beta \theta^{(t-1)} =& \theta^{(t)} - \beta \theta^{(t-2)} - \eta \nabla \mathcal{J} (\theta^{(t-1)}),\\
    \theta^{(t+1)} + p_{k+1} =&  \theta^{(t)} + p_{k} -  \frac{\eta}{1-\beta} \nabla \mathcal{J} (\theta^{(t-1)}),\quad p_k :=  \frac{\beta}{1-\beta}( \theta^{(t)}-\theta^{(t-2)} ). \\
\end{align}

Next, we have
\begin{align}
        &\Vert   \theta^{(t+1)} + p_{k+1} - \theta^* \Vert_{\ell_2}^2 \\
        =&  \Vert   \theta^{(t)} + p_{k} - \theta^* \Vert^2_{\ell_2} +(\frac{\eta}{1-\beta})^2 \Vert \nabla \mathcal{J} (\theta^{(t-1)})  \Vert^2_{\ell_2}-\frac{2\eta}{1-\beta} \langle   \theta^{(t)} + p_{k} - \theta^* ,  \nabla \mathcal{J} (\theta^{(t-1)})  \rangle \\
        =&  \Vert   \theta^{(t)} + p_{k} - \theta^* \Vert^2_{\ell_2} +(\frac{\eta}{1-\beta})^2 \Vert \nabla \mathcal{J} (\theta^{(t-1)})  \Vert^2_{\ell_2}\\
        &-\frac{2\eta (1+\beta)}{1-\beta} \langle   \theta^{(t)} -\theta^{(t-1)} ,  \nabla \mathcal{J} (\theta^{(t-1)})  \rangle  - \frac{2\eta \beta}{1-\beta} \langle   \theta^{(t-1)} -\theta^{(t-2)} ,  \nabla \mathcal{J} (\theta^{(t-1)})  \rangle - \frac{2\eta }{1-\beta} \langle   \theta^{(t-1)} -\theta^* ,  \nabla \mathcal{J} (\theta^{(t-1)})  \rangle \\
        =& \Vert   \theta^{(t)} + p_{k} - \theta^* \Vert^2_{\ell_2} +(\frac{\eta}{1-\beta})^2 \Vert \nabla \mathcal{J} (\theta^{(t-1)})  \Vert^2_{\ell_2} +  \frac{2\eta (1+\beta)}{1-\beta} \langle   \theta^{(t)} -\theta^{(t-1)} , \nabla \mathcal{J} (\theta^{(t)}) -  \nabla \mathcal{J} (\theta^{(t-1)}) \rangle\\
        & -\frac{2\eta (1+\beta)}{1-\beta} \langle   \theta^{(t)} -\theta^{(t-1)} , \nabla \mathcal{J} (\theta^{(t)})   \rangle  - \frac{2\eta \beta}{1-\beta} \langle   \theta^{(t-1)} -\theta^{(t-2)} ,  \nabla \mathcal{J} (\theta^{(t-1)})  \rangle - \frac{2\eta }{1-\beta} \langle   \theta^{(t-1)} -\theta^* ,  \nabla \mathcal{J} (\theta^{(t-1)})  \rangle .
\end{align}

By applying Eqs.~\eqref{eq: convex condition start}-\eqref{eq: convex condition end}, the terms in above equation have upper bounds as 
\begin{align}
   \frac{2\eta (1+\beta)}{1-\beta} \langle   \theta^{(t)} -\theta^{(t-1)} , \nabla \mathcal{J} (\theta^{(t)}) -  \nabla \mathcal{J} (\theta^{(t-1)}) \rangle  \leq &   \frac{\eta L (1+\beta)}{1-\beta}  \Vert  \theta^{(t)} -\theta^{(t-1)} \Vert^2_{\ell_2} \leq \frac{\eta^3 L(1+\beta)(1-\beta^\frac{t}{2})^2}{(1-\beta)^3}C^2,\\ 
   -\frac{2\eta (1+\beta)}{1-\beta} \langle   \theta^{(t)} -\theta^{(t-1)} , \nabla \mathcal{J} (\theta^{(t)})   \rangle  \leq & -\frac{2\eta (1+\beta)}{1-\beta} ( \mathcal{J} (\theta^{(t)}) - \mathcal{J} (\theta^{(t-1)}) ),\\ 
   - \frac{2\eta \beta}{1-\beta} \langle   \theta^{(t-1)} -\theta^{(t-2)} ,  \nabla \mathcal{J} (\theta^{(t-1)})  \rangle \leq &-\frac{2\eta \beta}{1-\beta} ( \mathcal{J} (\theta^{(t-1)}) - \mathcal{J} (\theta^{(t-2)} )),\\
   - \frac{2\eta }{1-\beta} \langle   \theta^{(t-1)} -\theta^* ,  \nabla \mathcal{J} (\theta^{(t-1)})  \rangle  \leq &  - \frac{2\eta }{1-\beta}  (\mathcal{J}(\theta^{(t-1)}) -\mathcal{J}(\theta^*) + \frac{1}{2L} \Vert\nabla \mathcal{J} (\theta^{(t-1)}) \Vert^2_{\ell_2} ).   \\
\end{align}

Therefore, we have
\begin{align}
&\Vert   \theta^{(t+1)} + p_{k+1} - \theta^* \Vert_{\ell_2}^2 -\Vert   \theta^{(t)} + p_{k} - \theta^* \Vert^2_{\ell_2}+ \frac{2\eta (1+\beta)}{1-\beta} ( \mathcal{J} (\theta^{(t)}) - \mathcal{J} (\theta^{(t-1)}) ) +     \frac{2\eta \beta}{1-\beta} ( \mathcal{J} (\theta^{(t-1)}) - \mathcal{J} (\theta^{(t-2)} ) )+ \frac{2\eta }{1-\beta}  (\mathcal{J}(\theta^{(t-1)}) -\mathcal{J}(\theta^*))  \\
\leq & \frac{\eta}{1-\beta} (\frac{\eta}{1-\beta} - \frac{1}{L}) \Vert \nabla \mathcal{J} (\theta^{(t-1)})  \Vert^2_{\ell_2}  +  \frac{\eta^3 L(1+\beta)(1-\beta^\frac{t}{2})^2}{(1-\beta)^3}C^2\\
\leq &\frac{\eta^3 L(1+\beta)(1-\beta^\frac{t}{2})^2}{(1-\beta)^3}C^2 \leq \frac{\eta^3 L(1+\beta)}{(1-\beta)^3}C^2.
\end{align}

By summing the two sides of above inequality, there is 
\begin{align}
    \frac{2\eta}{1-\beta} \frac{1}{T}\sum_{t=1}^T  (\mathcal{J}(\theta^{(t)}) -\mathcal{J}(\theta^*)) \leq &\frac{1}{T} \Big( \Vert   \theta^{(2)}  - \theta^* \Vert^2_{\ell_2} + \Vert   \theta^{(1)}  - \theta^* \Vert^2_{\ell_2}+ \frac{2\eta (1+\beta)}{1-\beta} ( \mathcal{J} (\theta^{(1)}) - \mathcal{J} (\theta^* ) +       \frac{2\eta \beta}{1-\beta} ( \mathcal{J} (\theta^{(0)}) - \mathcal{J} (\theta^* )  \Big) \\
    &+ \frac{\eta^3 L(1+\beta)}{(1-\beta)^3}C^2,
\end{align}
where $\theta^{(0)} $ can be regarded as  $\theta^{(0)} =  \theta^{(1)}  $. Therefore, the scaling can be given 
\begin{align}
    \frac{1}{T}\sum_{t=1}^T  (\mathcal{J}(\theta^{(t)}) -\mathcal{J}(\theta^*))  \sim \mathcal{O} (\frac{\eta^2 L(1+\beta)}{2(1-\beta)^2}C^2)  ,
\end{align}
which completes the proof.

\end{proof}

The lemma shows that even in the presence of delay, the optimizer maintains the same order of convergence rate as in the delay-free case for strongly convex problems, with an additional steady-state error that cannot be eliminated iteratively. Although the Ising problem is non-convex, many analyses of gradient optimizers evaluate behavior on strongly convex problems to infer properties in non-convex settings. Therefore, it can be expected that the asynchronous momentum optimizer maintains robustness against delays when applied to the non-convex Ising problem. Since steady-state errors do not affect stability, the asynchronous momentum optimizer combined with the CLF provides fast and stable convergence with the potential to escape local minima.

Finally, when the oscillator period is much larger than the delay, the control can be approximated as instantaneous. In such cases, a delay-free real-time controller can be designed, suitable for low-speed scenarios.

\begin{remark}[A special case without delay]

If the controller has an ideal computational speed, capable of processing signals in real-time and outputting signals, then the dynamics of AIM with a real-time controller is a one-order non-linear dynamical system such that
\begin{align}
  -  \nabla R(\psi)\cdot  \dot{\bm{\mu}}  + \Big(\nabla^2 K(\psi)  \cdot \nabla R(\psi) +\nabla^2 R(\psi) \cdot \nabla K(\psi) -\frac{1-\beta}{\sqrt{\kappa}}\nabla K(\psi)\Big) \cdot \bm{\mu}\\
  +\nabla^2 R(\psi) \cdot \nabla R(\psi) \cdot \bm{\mu}^2- \frac{1-\beta}{\sqrt{\kappa}}\nabla K(\psi) +  \nabla^2 K(\psi) \cdot\nabla K(\psi)+ \eta \nabla E^\prime (\psi ) = 0.   
\end{align}
\end{remark}
\begin{proof}
The momentum method without delay can be approximated as a nonlinear second-order system \cite{chen2024accelerated}

\begin{align}
\ddot{\psi} + \frac{1-\beta}{\sqrt{\kappa}}\dot{\psi} + \eta \nabla E^\prime (\psi ) = 0,
\end{align}
where $\kappa\in\mathbb{R}$ denotes a small step size. Combined with 
\begin{align}
\dot{\psi} = -\nabla K(\psi) - \bm{\mu} \cdot \nabla R(\psi), 
\end{align}
we have
\begin{align}
&-\nabla^2 K(\psi) \cdot \dot{\psi}  - \dot{\bm{\mu}}\cdot   \nabla R(\psi) -  \nabla^2 R(\psi) \cdot \dot{\psi} \cdot \bm{\mu} + \frac{1-\beta}{\sqrt{\kappa}}\dot{\psi} + \eta \nabla E^\prime (\psi ) = 0\\
\Rightarrow&    -\nabla^2 K(\psi) \cdot (-\nabla K(\psi) - \bm{\mu} \cdot \nabla R(\psi))  - \dot{\bm{\mu}}\cdot   \nabla R(\psi) -  \nabla^2 R(\psi) \cdot  (-\nabla K(\psi) - \bm{\mu} \cdot \nabla R(\psi)) \cdot \bm{\mu}\\
& + \frac{1-\beta}{\sqrt{\kappa}} (-\nabla K(\psi) - \bm{\mu} \cdot \nabla R(\psi)) + \eta \nabla E^\prime(\psi ) = 0 \\
\Rightarrow&  -  \nabla R(\psi)\cdot  \dot{\bm{\mu}}  + \Big(\nabla^2 K(\psi)  \cdot \nabla R(\psi) +\nabla^2 R(\psi) \cdot \nabla K(\psi) -\frac{1-\beta}{\sqrt{\kappa}}\nabla K(\psi)\Big) \cdot \bm{\mu}+\nabla^2 R(\psi) \cdot \nabla R(\psi) \cdot \bm{\mu}^2\\
&- \frac{1-\beta}{\sqrt{\kappa}}\nabla K(\psi) +  \nabla^2 K(\psi) \cdot\nabla K(\psi)+ \eta \nabla E^\prime (\psi ) = 0.
\end{align}
\end{proof}

\section{Numerical simulation on CAIMs}

To validate the broad generality of CAIM across different AIM implementations, numerical simulations are conducted on three representative architectures: OIM, BRIM, and ROSC. In this section, multiple properties of CAIM are analyzed jointly, including the effects of injection strength, time delay, initial conditions, and noise. These factors are closely related to the solution accuracy and convergence speed of CAIM, and their investigation provides deeper insight into CAIM behavior and guidance for selecting appropriate parameter configurations in practical applications. CAIM notations are assigned as COIM, CBRIM and CROSC, correspondingly.

The benchmark problem used in the numerical experiments is the SpinModel \cite{Wang2025Parallel}, which is equivalent to a weighted fully connected MaxCut problem, with parameters $J_{ij, i<j}, h_i \in \mathcal{U} \{ 0, \pm 0.1,\pm 0.2, \cdots \pm 1.0\}$. It is worth noting that although $J_{ij}$ and $h_i$ are drawn from the same distribution, the actual coupling strength of $J_{ij}$ is twice that of $h_i$ due to the symmetry condition $J_{ij} = J_{ji}$. For each experiment, the SpinModel instance is randomly generated 20 times. In the numerical simulations, running time is normalized to phase time, and conversion to physical time depends on hardware-specific parameters. The sliding-window duration used to determine convergence in Eq. \eqref{eq: convergence detection} is set to $\varepsilon = 0.025$. Additionally, since the function $y=x$ in BRIM is unbounded, it is approximated by $x\approx 10\tanh(0.1x)$ to facilitate periodic behavior.

\subsection{Injection amplitude (binarization strength)}
\begin{figure}
    \centering
    \includegraphics[width=\textwidth]{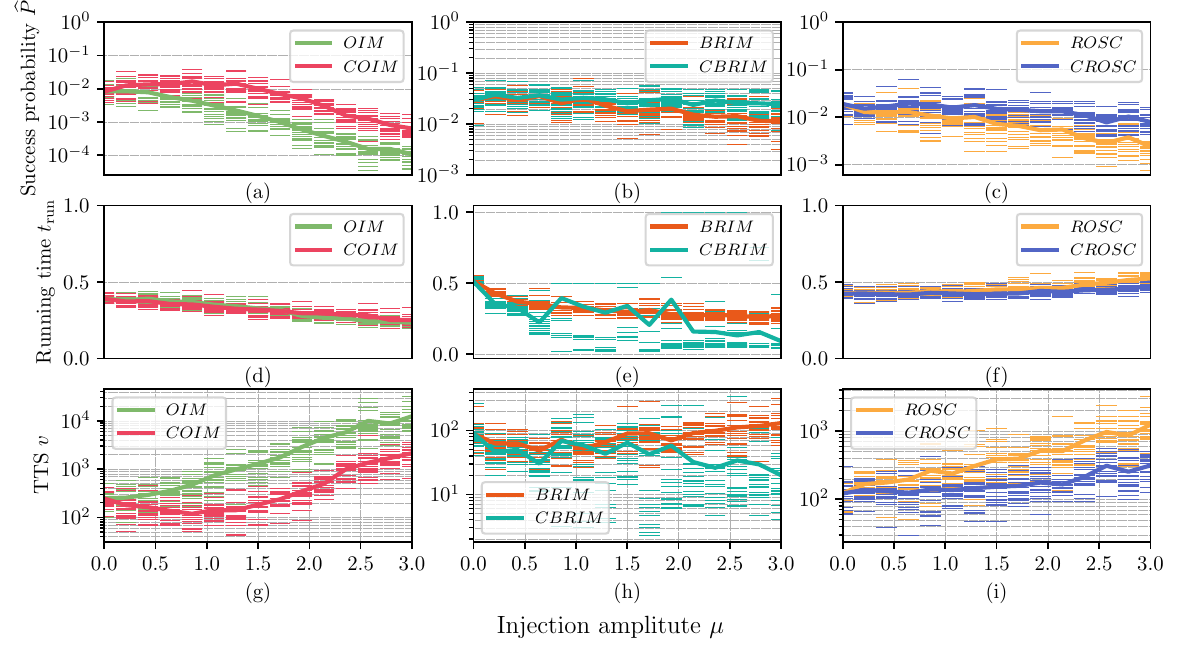}
    \caption{Comparisons on the success probability and running time between CAIM and AIM. }\label{si-fig: comparison on accuracy}
\end{figure}

The binarization strength $\mu$, also referred to as the injection amplitude, directly determines the upper bound of the solution accuracy achievable by AIM. As analyzed earlier, an excessively weak binarization strength degrades correctness, while an overly strong binarization strength reduces reachability. The numerical simulation results shown in Fig. \ref{si-fig: comparison on accuracy} corroborate this conclusion. In the simulations, $\mu$ is varied within the range $0\leq \mu \leq3$. It can be clearly observed that OIM exhibits the most pronounced behavior, where the success probability first increases and then decreases, indicating a clear upper bound. Moreover, the running time of both OIM and BRIM decreases as $\mu$ increases, which implies that under large $\mu$ the system converges prematurely. ROSC shows relatively stable behavior because the gradient of the $\tanh$ function approaches zero away from the origin, leading to slower energy variation, and therefore it is less likely to satisfy the convergence criterion in Eq. \eqref{eq: convergence detection}. In addition, it is evident that CAIM achieves significant improvements over AIM in both success probability and running time, with the success probability increasing by approximately a factor of two on average and the running time accelerating by about $10\%$ on average, while the TTS can even be improved by an order of magnitude. The improvement of CAIM in terms of TTS is particularly prominent when the binarization constraint is overly strong, further validating the performance gains introduced by CAIM.

\subsection{Time delay}
\begin{figure}
    \centering
    \includegraphics[width=\textwidth]{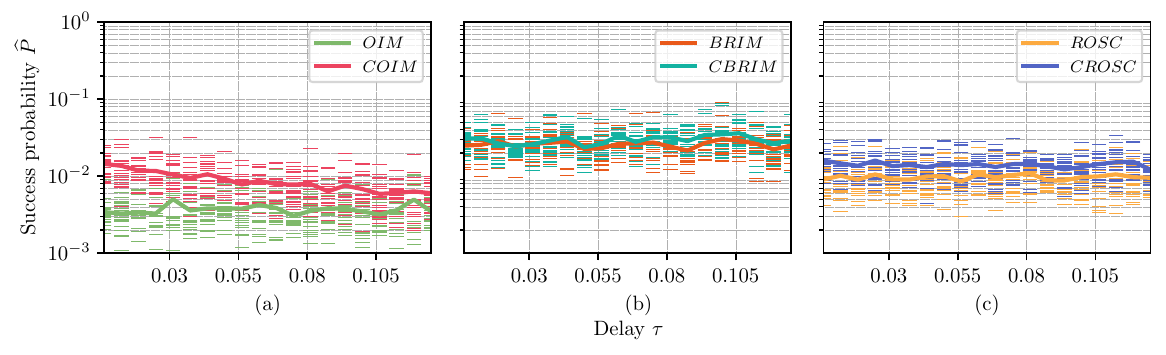}
    \caption{Effect of time delay $\tau$ on success probability of CAIMs. }\label{si-fig: time delay simulation}
\end{figure}

The time delay $\tau$ is an important factor influencing CAIM performance and is mainly determined by the computational latency of the controller or hardware delays. As discussed previously, excessive delay can lead to instability in the control system, and the numerical simulation results in Fig. \ref{si-fig: time delay simulation} verify this conclusion. Compared with AIM as a baseline, the success probability of CAIM decreases under excessively large delays, with the effect being particularly evident for COIM and CROSC. This indicates that, in hardware implementations of the controller, response speed must be carefully considered, since a mismatch with the high-speed oscillatory circuits of AIM would significantly undermine the performance gains provided by CAIM.

\subsection{Initial value}

\begin{figure}
    \centering
    \includegraphics[width=\textwidth]{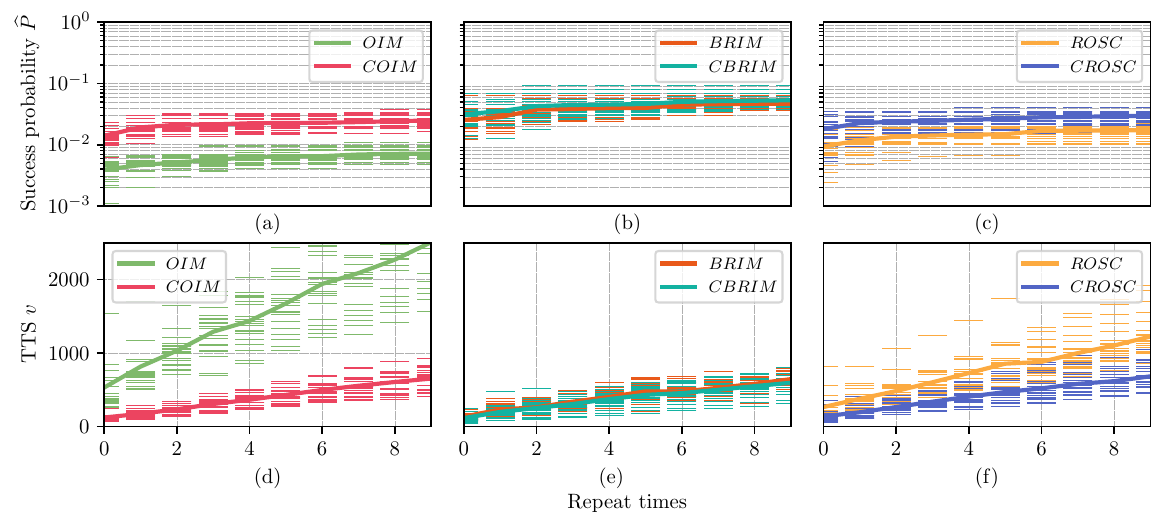}
    \caption{Performance with different initial state generation times. }\label{si-fig: intial value simulation}
\end{figure}

The initial state of AIM serves as the starting point of its dynamical evolution and typically determines which local minimum basin the system eventually settles into. In practice, Ising machines, including QA, CIM, DIM, AIM, and simulated annealing algorithms, usually attempt multiple runs with different initial conditions and select the best result when solving a given Ising problem. The numerical simulations in Fig. \ref{si-fig: intial value simulation} also incorporate this consideration, where each run corresponds to a random initialization for the same problem, repeated multiple times to obtain the optimal solution. The simulation results show that both AIM and CAIM benefit from multiple initializations, leading to increased success probabilities; however, although the success probability improves, the running time also increases proportionally. For the TTS metric, which accounts for both success probability and running time, the time cost of COIM and CROSC appears relatively high. Overall, CAIM still maintains an advantage in success probability, and the specific number of repetitions can be chosen according to accuracy-oriented or efficiency-oriented application scenarios.

It is worth noting that, in the numerical experiments, the random initial states are chosen as $\phi_i \sim \mathcal{U}(0,2\pi)$ for OIM, $\phi_i \sim \mathcal{U}(-10,10)$ for BRIM, and $\phi_i \sim \mathcal{U}(-1,1)$ for ROSC. Although truly random initial states are often difficult to realize precisely in circuit hardware implementations of AIM, CAIM can theoretically realize any equivalent random initial state, since this is equivalent to placing the system state at a specified configuration, which can be achieved via impulse functions as described in Remark \ref{si-rmk: contrallability}. Therefore, if the controller is digital, it can randomly generate a set of state values and then stabilize the system at the target state by adjusting the control pulses $\bm{\mu}(t)$ before starting the solving process, thereby equivalently realizing random initialization.

\subsection{Noise}

\begin{figure}
    \centering
    \includegraphics[width=\textwidth]{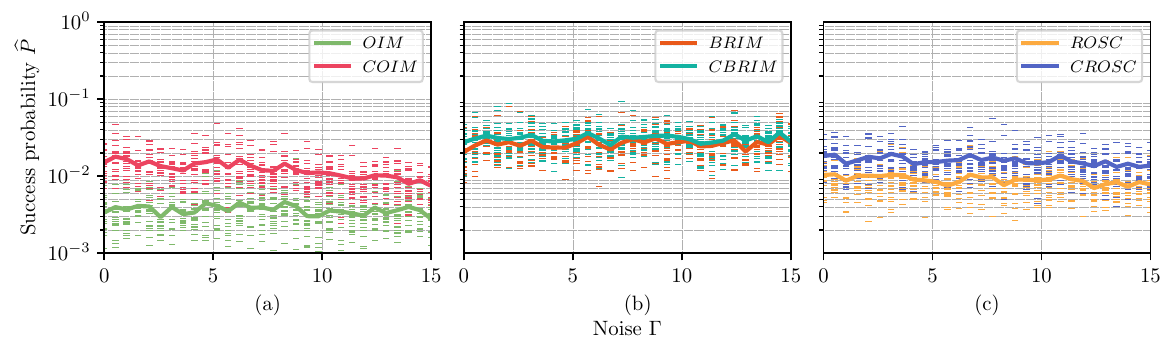}
    \caption{Performance variation with noises. }\label{si-fig: noise simulation}
\end{figure}

In this experiment, noise is introduced into the dynamical evolution of the system. It is assumed that the noise is isotropic across all oscillators, meaning that each oscillator experiences noise of identical strength, with $\Gamma$ determining the noise intensity, and the noise is assumed to be Gaussian white noise $\xi_i\sim \mathcal{N}(0,1^2)$. Accordingly, the dynamics can be written as 
\begin{align}
 \rmd \phi_i = \Big( -\frac{\partial K }{ \partial \phi_i} - \mu_i (t)\frac{\partial R}{ \partial \phi_i}~ \Big) \rmd  t + \Gamma \xi_i, \quad    \xi_i\sim \mathcal{N}(0,1^2).
\end{align}
The simulation results shown in Fig. \ref{si-fig: noise simulation} indicate that CBRIM is insensitive to noise, while CROSC and ROSC exhibit similar trends with a slight degradation, and COIM shows a pronounced decline in success probability. This behavior is mainly attributed to the fact that both BRIM and ROSC employ $\tanh$-shaped functions, which have small gradients near the periodic boundaries and thus suppress oscillations, rendering them less sensitive to noise. In contrast, the cosine function used in COIM sustains oscillations and preserves the influence of noise, while the momentum-based optimizer further amplifies noise effects, resulting in a noticeable reduction in success probability.

\end{document}